\documentclass[a4paper]{amsart}
\usepackage[text={5.2in,8in},centering]{geometry}
\usepackage[backrefs]{amsrefs}
\usepackage{amssymb}
\usepackage{enumerate}
\usepackage{graphicx}
\numberwithin{equation}{section}
\theoremstyle{plain}
\newtheorem{theorem}{Theorem}[section]

\newtheorem{corollary}[theorem]{Corollary}
\newtheorem{lemma}[theorem]{Lemma}
\theoremstyle{remark}
\newtheorem{remark}[theorem]{Remark}

\newtheorem*{quest*}{Question}
\newtheorem*{remark*}{Remark}
\theoremstyle{definition}
\newtheorem{definition}[theorem]{Definition}
\newtheorem*{definition*}{Definition}
\newtheorem*{notation*}{Notation}
\newtheorem*{notations*}{Notations}
\providecommand{\B}{\mathbf}
\providecommand{\BS}[1]{\boldsymbol{#1}}
\providecommand{\C}{\mathcal}
\providecommand{\D}{\mathbb}
\providecommand{\F}[1]{\mathfrak{#1}}
\providecommand{\R}{\mathrm}
\newcommand{\ee}{\mathrm{e}}
\newcommand{\eul}{\mathrm{e}}

\newcommand{\coloneq}{\mathrel{\mathop:\!\!=}}%---------------------------------------------------------------------%
\providecommand{\abs}[1]{\lvert#1\rvert}
\providecommand{\accol}[1]{\lbrace#1\rbrace}
\providecommand{\croch}[1]{\left\lceil#1\right\rceil}
\providecommand{\norm}[1]{\lVert#1\rVert}
\newcommand{\DSmkn}{\B{DS}\BS{(}m,p,k,I,n\BS{)}}
\newcommand{\DSmkN}{\B{DS}\BS{(}m,p,k,I,N\BS{)}}
\newcommand{\DSmkponen}{\B{DS}\BS{(}m,p,k+1,I,n\BS{)}}
\newcommand{\DSmkponeN}{\B{DS}\BS{(}m,p,k+1,I,N\BS{)}}

\newcommand{\DSmzeron}{\B{DS}\BS{(}m,p,0,I,n\BS{)}}

\newcommand{\FISk}{\B{FIS}\BS{(}k,p,N\BS{)}}
\newcommand{\FISkone}{\B{FIS}\BS{(}k+1,p,N\BS{)}}

\newcommand{\Wone}[1]{\B{W1}(#1)}
\newcommand{\Wtwo}[1]{\B{W2}(#1)}

\newcommand{\EOne}{\textbf{(E\,1)}}
\newcommand{\ETwo}{\textbf{(E\,2)}}
\newcommand{\EThree}{\textbf{(E\,3)}}
\newcommand{\EFour}{\textbf{(E\,4)}}

\newcommand{\IOne}{\textbf{(I\,1)}}
\newcommand{\ITwo}{\textbf{(I\,2)}}

\DeclareMathOperator{\boxl}{\BS{\varLambda}}
\DeclareMathOperator{\boxx}{\mathbf{B}}
\DeclareMathOperator{\cell}{\mathbf{C}}
\DeclareMathOperator{\compl}{c}

\DeclareMathOperator{\diam}{diam}
\DeclareMathOperator{\dist}{dist}
\DeclareMathOperator{\supp}{supp}
\DeclareMathOperator*{\esinf}{ess\,inf}
\DeclareMathOperator*{\essup}{ess\,sup}
\DeclareMathOperator{\expect}{\mathbb{E}}

\DeclareMathOperator{\green}{G}
\DeclareMathOperator{\Green}{\mathbf{G}}
\DeclareMathOperator{\intr}{int}

\DeclareMathOperator{\nonsing}{NS}
\DeclareMathOperator{\one}{\mathbf{1}}

\DeclareMathOperator{\out}{out}
\DeclareMathOperator{\prob}{\mathbb{P}}
\DeclareMathOperator{\sing}{S}

\def\half{{\frac{1}{2}}}
\def\spec{{\sigma}}
\begin{document}
%---------------------------------------------------------------------%
\title[Anderson localization for a multi-particle alloy-type model]
{Anderson localization for a multi-particle model\\
with alloy-type external potential}
%---------------------------------------------------------------------%
\author[A. Boutet de Monvel]{Anne Boutet de Monvel$^1$}
\author[V. Chulaevsky]{Victor Chulaevsky$^2$}
\author[P. Stollmann]{Peter Stollmann$^3$}
\author[Y. Suhov]{Yuri Suhov$^4$}
%---------------------------------------------------------------------%
\address{$^1$Institut de Math\'ematiques de Jussieu\\
Universit\'e Paris Diderot Paris 7\\
175 rue du Chevaleret, 75013 Paris, France\\
E-mail: aboutet@math.jussieu.fr}
\address{$^2$D\'epartement de Math\'ematiques\\
Universit\'{e} de Reims, Moulin de la Housse, B.P. 1039,\\
51687 Reims Cedex 2, France\\
E-mail: victor.tchoulaevski@univ-reims.fr}

\address{$^3$Fakult\"at f\"ur Mathematik\\
Technische Universit\"at Chemnitz\\
09107 Chemnitz, Germany\\
E-mail: peter.stollmann@mathematik.tu-chemnitz.de}

\address{$^4$Statistical Laboratory, DPMMS\\
University of Cambridge, Wilberforce Road, \\
Cambidge CB3 0WB, UK\\
E-mail: Y.M.Suhov@statslab.cam.ac.uk}
\date{\today}
%---------------------------------------------------------------------%
\begin{abstract}
We establish exponential localization for a multi-particle Anderson model
in a Euclidean space $\D{R}^d$, $d\geq 1$, in presence of a non-trivial
short-range interaction and an alloy-type random external potential.
Specifically, we prove that all eigenfunctions with eigenvalues near
the lower edge of the spectrum decay exponentially.
\end{abstract}
%---------------------------------------------------------------------%
\maketitle
%---------------------------------------------------------------------%
%:sec.1
%---------------------------------------------------------------------%
\section{Introduction. The $N$-particle Hamiltonian in the continuum} \label{sec:intro}
%---------------------------------------------------------------------%
%:s.1.1

%---------------------------------------------------------------------%
\subsection{The model}

This paper considers an $N$-particle Anderson model in $\D{R}^d$ with
interaction.
The Hamiltonian $\B{H}=\B{H}^{(N)}(\omega)$ is a random Schr\"odinger operator of the form
\begin{equation}        \label{eq:random.hamiltonian}    %\eqno(1.1)
\B{H}^{(N)}(\omega)=-\frac{1}{2}\B{\Delta}+\B{U}+\B{V}(\omega )
\end{equation}
acting on functions from  $L^2(\D{R}^d\times\dots\times\D{R}^d)\simeq
L^2(\D{R}^d)^{\otimes N}$. This means that we consider $N$ quantum particles in $\D{R}^d$. The joint position vector is
$\B{x}=(x_1,\dots,x_N)\in\D{R}^{Nd}$, where component
$x_j=(\R{x}_j^{(1)},\dots,\R{x}_j^{(d)})\in\D{R}^d$ represents
the $j$-th particle, $j=1,\dots,N$. Next,
\[
-\frac{1}{2}\B{\Delta}=-\frac{1}{2}\sum_{1\leq j\leq N}\varDelta_j
\]
is the standard kinetic energy operator obtained by adding up the kinetic
energies $-\frac{1}{2}\varDelta_j$ of the individual particles; here, $\varDelta_j$ denotes the $d$-dimensional
Laplacian.

The interaction energy operator is denoted by $\B{U}$: it is the operator of
multiplication by a function $\D{R}^{Nd}\ni\B{x}\mapsto U(\B{x})$,
the inter-particle potential (which can also incorporate a deterministic
external potential). Finally, $\B{V}(\omega)$ is the operator of multiplication
by a function
\begin{equation}            \label{eq:external.field}        %\eqno(1.2)
\B{x}=(x_1,\dots,x_N)\in\D{R}^{Nd}\longmapsto V(x_1;\omega)+\dots+V(x_N;\omega),
\end{equation}
where $V: \D{R}^d\times \Omega \to\D{R}$ is the random external field potential, relative to a probability space $(\Omega, \C{F}, \D{P})$, acting on an individual particle.

Assumptions on $U(\B{x})$ and $V(x;\omega)$ are discussed below,
in subsections
\ref{ssec:interaction} and \ref{ssec:external}. In essence, $U$ is required
to be a sum of short-range inter-particle potentials while $V$ is assumed to
be of the so-called alloy type. We refer to the quantum system with Hamiltonian
$\B{H}$ as a multi-particle alloy-type Anderson model in $\D{R}^d$.

In this paper, we analyse spectral properties of $\B{H}$ by
using the method called Multi-Scale Analysis (MSA), more precisely, a
multi-particle adaptation of a single-particle ``continuous-space'' version
of the MSA. Our main result is Theorem \ref{thm:main}, asserting that with probability one the
spectrum of operator $\B{H}$ near
its lower edge is pure point, with an exponential decay of the
corresponding eigenfunctions. Such a phenomenon is known as
(exponential) Anderson localisation. In the context of the
alloy-type Anderson models one often refers to the famous
``Lifshits-tail'' picture suggesting possible localisation domains
in terms of relevant parameters.

The fact that the spectrum near its lower edge is non-empty (and even dense) follows 
easily from the assumption that the interaction potential $U$ has a short range, combined with 
known facts about spectra of single-particle Anderson-type Hamiltonians. 

Theorem \ref{thm:main} is the first rigorous result on localisation
in multi-particle continuous-space Anderson models.

For lattice (tight-binding) Anderson models, the multi-particle adaptation
of the MSA has been developed in earlier papers \cite{CS08}, \cite{CS09a},
\cite{CS09b}. An alternative approach based on the Fractional Moment Method (FMM)
was successfully employed, for multi-particle lattice Anderson models, in
\cite{AW09a}; see also \cite{AW09b}.

The structure of the present paper is commented on in subsection
\ref{ssec:structure}.

%---------------------------------------------------------------------%
%:s.1.2
%---------------------------------------------------------------------%
\subsection{Basic notation}     \label{ssec:notations}

Throughout this paper, we fix integers $N>1$ and $d\ge 1$ (which can be arbitrary) and work with configurations of $n\le N$ distinguishable quantum particles in $\D{R}^d$. The configuration space of an $n$-particle system is the Euclidean space $\left(\D{R}^{d}\right)^n$ which is canonically identified with $\D{R}^{nd}$. A similar identification is always used for the cubic lattices:
$\left(\D{Z}^{d}\right)^n \cong \D{Z}^{nd}$.

It is convenient to endow $\D{R}^d$ and $\D{R}^{nd}$ with $\max$-norm:
\begin{equation}                              %\eqno (1.4)
\abs{x}=\max_{1\leq i\leq d}\abs{\R{x}^{(i)}},\;\;
\abs{\B{x}}=\max_{1\leq j\leq n}\abs{x_j}.
\end{equation}
The distance ``$\dist$'' below is induced by this norm.
In terms of the max-norm  in $\D{R}^d$ the ball of radius $L$
centered at $u = (u^{(1)}, \ldots, u^{(d)}$ is the the cube
$$
\varLambda_L(u) :=
{\operatornamewithlimits{ {\times}}\limits_{i=1}^d}
\Big(\R{u}^{(i)}-L,\R{u}^{(i)}+L\Big)\subset\D{R}^d
$$
and the ball in $\D{R}^{nd}$ of radius $L$ centered at
$\B{u} = (u_1, \ldots, u_N)$ is the cube
\begin{equation}                               %\eqno (1.3)
\boxl_L(\B{u})=
{\operatornamewithlimits{\times}\limits_{j=1}^n}
\varLambda_L(u_j)\subset\D{R}^{nd}.
\end{equation}
Sometimes we will
use the symbol $\boxl_L^{(n)}(\B{u})$ to put emphasis on the
number of particles
in the system.  For our purposes, it suffices
to consider only cubes centered at lattice points  $u\in\D{Z}^d$ and
$\B{u}\in\D{Z}^{nd}$.
For that reason, letters $u,v,w$ and $\B{u},\B{v},\B{w}$ will
always refer to points in the corresponding lattices.

We denote  by $\B{1}_{\B{A}}$
the characteristic function of a set $\B{A}\subset\D{R}^{nd}$
and also, with a standard abuse of notation, the operator of multiplication
by this function.

We also need ``lattice cubes'':
\begin{equation}                              %\eqno (1.5)
B_L(u)=\varLambda_L(u)\cap\,\D{Z}^d,\;\;
\B{B}_L(\B{u})=\boxl_L(\B{u})\cap\,\D{Z}^{nd},
\end{equation}
and ``unit cells'', or simply ``cells'':
\begin{equation}                              %\eqno (1.6)
C(u)=\varLambda_1(u)\subset\D{R}^d,\;\;
\B{C}(\B{u})=\boxl_1(\B{u})\subset\D{R}^{nd}.
\end{equation}
In what follows, all these sets are often called ``boxes'',
single-particle boxes for $\varLambda_L(u)$, $B_L(u)$ and $C(u)$
and $n$-particle boxes for $\boxl_L(\B{u})$, $\B{B}_L(\B{u})$ and $\B{C}(\B{u})$.
A ``cellular set'' is a finite union of cells.

We will also use "annular" sets, or shortly, annuli,
defined as the difference $\boxl_{L+\R{w}}(\B{u})\setminus
\boxl_L(\B{u})$ where $\R{w}>0$ is the width and $\B{u}$
the centre.

%---------------------------------------------------------------------%
%:s.1.3
%---------------------------------------------------------------------%
\subsection{Interaction potential}  \label{ssec:interaction}

The interaction potential $U$ is of the form
\begin{equation}  \label{eq:interaction}       %\eqno(1.7)
U(\B{x})=\sum_{k=1}^N\sum_{1\leq i_1<\dots<i_k\leq N}
\varPhi^{(k)}(x_{i_1},\dots,x_{i_k})
\end{equation}
where $\B{x}=(x_1,\dots,x_N)\in\D{R}^{Nd}$. The functions $\varPhi^{(k)}
\colon\D{R}^{d\times k}\to\D{R}$ are $k$-body interaction potentials,
$k=1,\dots,N$,
satisfying the following properties \IOne-\ITwo, for $k=1,\dots,N$ and
$\B{y}=(y_1,\dots,y_k)\in\D{R}^{d\times k}$:
\begin{enumerate} [\IOne]
\item
\emph{Boundedness and nonnegativity}: There exists a
constant $\R{u}_0\in(0,+\infty)$ such that
\begin{equation}           \label{eq:interaction.boundedness}
0\leq\varPhi^{(k)}(\B{y})\leq\R{u}_0.
\end{equation}
\end{enumerate}

\begin{enumerate} [\ITwo]
\item
\emph{Finite range}: For some constant $\R{r}_0\in(0,+\infty)$ and
for $k=2,\dots,N$,
\begin{equation}           \label{eq:interaction.finite.range}
\max_{1\leq i\leq k}\min_{j\neq i}\,
\abs{y_i-y_j}\geq\R{r}_0\implies\varPhi^{(k)}(\B{y})=0.
\end{equation}
\end{enumerate}

%-----------------------%
\begin{remark}             \label{rem:interaction.nonnegativity}
The non-negativity of the potentials $\varPhi^{(k)}$ is
used to simplify the statement of the main result
(see Theorem~\ref{thm:main} below) and shorten the proof of technical
assertions.

We can also relax the boundedness condition, by allowing ``hard-core
potentials'', such that, for any $k=2,\dots,N$, and for $0<\R{r}_1<\R{r}_0$,
\[
\min_{1\leq i<j\leq k}\abs{y_i-y_j}<\R{r}_1\implies\varPhi^{(k)}(\B{y})
=+\infty.\]
While symmetry of the interaction is not important for our methods,
it is usually assumed in physical applications.

On the other hand, the finite-range condition is essential.
Extending Theorem~\ref{thm:main} to the case of
infinite-range potentials seems an important and challenging problem.
\end{remark}
%-----------------------%

%---------------------------------------------------------------------%
%:s.1.4
%---------------------------------------------------------------------%
\subsection{External field potential}  \label{ssec:external}

As mentioned before, the random external potential $V(x;\omega)$,
$x\in\D{R}^d$, $\omega\in\varOmega$, is assumed to be of alloy-type,
over a cubic lattice. That is,
\begin{equation}        \label{eq:external}                 %\eqno(1.9)
V(x;\omega)=\sum_{s\in\D{Z}^d}\R{V}_s(\omega)\varphi_s(x-s).
\end{equation}
Here $\{\R{V}_s\}_{s\in\D{Z}^d}$, is a family of IID (independent,
identically distributed) real
random variables $V_s$ on some probability space
$(\varOmega,\F{B},\prob)$ and $\{\varphi_s\}_{s\in\D{Z}^d}$ is a
(nonrandom) collection of ``bump'' functions (not necessarily identical)
\[
\D{R}^d\ni y\mapsto\varphi_s(y).
\]
In probabilistic terms, $V$ is a real-valued random field (RF) on $\D{Z}^d$.
Physically speaking, random variable $\R{V}_s$ represents the amplitude of
the ``impurity'' at the site $s\in\D{Z}^d$ while function $\varphi_s$
describes the ``propagation'' of the impact of this impurity across $\D{R}^d$.

%---------------------------------------------------------------------%
\subsubsection{}
We assume the following conditions \EOne-\ETwo.
\begin{enumerate}[\EOne]
\item
\emph{Boundedness and nonnegativity}:
\begin{equation}             \label{eq:external.boundedness}   %\eqno(1.10)
\essup\;\R{V}_s<\infty,\;\;\esinf\;\R{V}_s= 0.
\end{equation}
\end{enumerate}

%-----------------------%
\begin{remark}             \label{rem:external.boundedness}
Again, the nonnegativity plays a technical role and is not crucial for
the main result. The boundedness condition  can be replaced by
finiteness of moments $\expect\abs{\R{V}_s}^{\R{r}}$ for some $\R{r}>0$.
\end{remark}
%-----------------------%

Let $s\in\D{Z}^d$ be a given site. Consider the distribution
function:
\begin{equation}   \label{eq:conditional.distribution}   %\eqno(1.11)
F(\R{y}):=\prob(\R{V}_s<\R{y}), \quad \R{y}\in\D{R},
\end{equation}
Condition (E1) implies that $F(\R{y})=0$ for $\R{y}<0$ and  $F(\R{y})=1$ for $y$ large enough.
\begin{enumerate}[\ETwo]
\item
\emph{Uniform H\"older-continuity of $F(\R{y}\vert\F{B}_s^{\compl})$}:
There exist  constants $a,\,b>0$ such that for all $\epsilon\in(0,1)$,
\begin{equation}           \label{eq:external.holder}        %\eqno(1.12)
\nu(\epsilon):=\sup_{s\in\D{Z}^d}\sup_{\R{y}\in\D{R}}\bigl\lbrack
F(\R{y}+\epsilon )
-F(\R{y})\bigr\rbrack\leq a\epsilon^b.
\end{equation}
\end{enumerate}

%-----------------------%
\begin{remark}             \label{rem:external.assumptions}
The main result of this paper remains
valid under a weaker assumption of log-H\"older continuity:
$\nu(\epsilon)\leq a\abs{\ln\epsilon}^{-b}$, for $b>0$ large enough.
\end{remark}
%-----------------------%

%---------------------------------------------------------------------%
\subsubsection{}
Lastly, we require two more conditions, \EThree-\EFour, \, on bump functions
$\varphi_s$.
\begin{enumerate}[\EThree]
\item
\emph{Boundedness, nonnegativity and compact support of $\varphi_s$}:
Functions $\varphi_s$ are nonnegative and have a compact support:
$\diam \left(\supp \varphi_s \right) \le R$, so that

\begin{equation}       \label{eq:bounded.bumps}      %\eqno(1.13)
\sup_{x\in\D{R}^d}\sum_{s\in\D{Z}^d}\;\varphi_s(x-s)<+\infty,
\end{equation}

\end{enumerate}
\begin{enumerate}[\EFour]
\item
\emph{Covering condition for $\varphi_s$}: For all $L\geq 1$, $u\in\D{R}^d$ and $x\in\varLambda_L(u)$,
\begin{equation}       \label{eq:covering.condition}        %\eqno(1.15)
\sum_{s\in\varLambda_L(u)\cap\,\D{Z}^d}\;\varphi_s(x-s)\geq 1.
\end{equation}
\end{enumerate}

%-----------------------%
\begin{remark}             \label{rem:bump.assumptions}
As above, assumptions \EThree-\EFour \, can be relaxed.
\end{remark}
%-----------------------%

From now on we assume that values $d\geq 1$ and $N>1$ are fixed.
We will work with fixed interaction potentials
$\varPhi^{(k)}$
in Eqn \eqref{eq:interaction}, $1\leq k\leq N$, a fixed
collection of bump functions $\varphi_s$ from Eqn \eqref{eq:external}
and a fixed
distribution function $F$ in Eqn \eqref{eq:conditional.distribution},
assuming the conditions \IOne-\ITwo$\,$ and \EOne-\EFour.

%---------------------------------------------------------------------%
%:s.1.5
%---------------------------------------------------------------------%
\subsection{Main result}       \label{ssec:main.result}

Under conditions \IOne-\ITwo$\,$ and \EOne-\EFour, operator $\B{H}^{(N)}(\omega )$
is correctly defined for $\prob$-almost all
$\omega\in\Omega$ (as a unique self-adjoint extension
from the set of
$C^2$-functions $f (\B{x})$ with compact support). Furthermore, the
(nonrandom) operator
$$\B{H}_0^{(N)}=-\frac{1}{2}\B{\Delta}+\B{U}$$
is also correctly defined and has the lower edge of its spectrum
at a point $E^0\geq 0$.

%-----------------------%
\begin{theorem}                     \label{thm:main}
Let $\B{H}^{(N)}(\omega)$ be the random operator defined
in Eqn \eqref{eq:random.hamiltonian}.
Then $\exists$ nonrandom constants $\eta^*>0$ and $m^*>0$ such that,
with $\prob$-probability one,

\begin{enumerate}[\rm(i)]
\item
The spectrum of $\B{H}^{(N)}(\omega)$ in $[E^0,E^0+\eta^*]$
is non-empty and pure point.
\item
All
eigenfunctions $\BS{\varPsi}_j(\B{x};\omega)$ of $\B{H}^{(N)}(\omega)$ with
eigenvalues $E_j(\omega)\in [E^0,E^0+\eta^*]$ satisfy exponential bounds
\begin{equation}  \label{eq:exponential.decay}
\left\|{\mathbf 1}_{\B{C}(\B{u})}\BS{\varPsi}_j(\;\cdot\;;\omega)\right\|_{L_2
(\D{R}^{Nd})}
\leq c_j(\omega)\eul^{-m^*\abs{\B{u}}},\;\;\B{u}\in\D{Z}^{Nd},
\end{equation}
where $c_j(\omega )\in (0,+\infty)$ are random constants.
\end{enumerate}
\end{theorem}
%-----------------------%

A direct application of general results on local regularity of (generalized) eigenfunctions of Schr\"{o}dinger operators, (see, e.g., \cite{CFKS}), gives rise to the following
\begin{corollary}\label{cor:main}
The eigenfunctions
$\BS{\varPsi}_j(\B{x};\omega)$ with eigenvalues $E^0\leq E_j(\omega )\leq
E^0+\eta^*$ satisfy the bounds:
\begin{equation}  \label{eq:exponential.decay.pointwise}
\left|\BS{\varPsi}_j(\B{x};\omega)
\right|
\leq {\widetilde c}_j(\omega)\,\eul^{-{\widetilde m}^*\abs{\B{x}}},
\;\;\B{x}\in\D{R}^{Nd},
\end{equation}
with  ${\widetilde m}^*>0$ and random constants
${\widetilde c}_j(\omega )\in (0,+\infty)$.
\end{corollary}

%-----------------------%

\begin{remark}             \label{rem:main.generalization}
Theorem \ref{thm:main} addresses the spectrum of $\B{H}^{(N)}(\omega )$ in
the whole Hilbert space $L_2(\D{R}^{Nd})$. This, of course, covers
subspaces $L_2^{\rm{sym}}(\D{R}^{Nd})$ and
$L_2^{\rm{asym}}(\D{R}^{Nd})$ formed by symmetric and antisymmetric
functions (bosonic and fermionic subspaces, respectively).

Next, as explained in Section \ref{sec:MSA}, when we increase the number
of particles $N$, keeping fixed the structure of potentials $U$ and $V$,
the width $\eta^*\rightarrow 0$.
A similar phenomenon is observed (due to essentially similar reasons)
when one uses the MSA to prove pure point spectrum in a single-particle
Anderson model with growing dimension $d$ of the one-particle
configuration space $\D{R}^d$. In fact, key estimates of the
single-particle MSA cannot be made uniform in $d$, in the framework
of existing technical tools.
On the other hand, it is possiblle to modify the argument
presented in
this paper and show that if the external random potential field
has the form  $gV(x,\omega)$, $g>0$ being a coupling amplitude,
then the width $\eta^*$ can be made of order $O(g)$, for any given
value of $N$.
\end{remark}
%-----------------------%

%---------------------------------------------------------------------%
\subsection{From MSA bounds to dynamical localization}       \label{ssec:MSA.to.DL}

The derivation of the dynamical localization from sufficiently strong MSA bounds in the framework of single-particle Anderson models, on a lattice or in a Euclidean space, is well-understood by now. For the multi-particle model considered in this paper, the derivation of dynamical localization from the key MSA bounds proven below in sections \ref{sec:MSA}--\ref{sec:caseIII} requires only a few minor modifications of known techniques, essentially of geometrical nature. We plan to publish it in a separate paper, in order to keep the size of the present manuscript within reasonable limits.

%---------------------------------------------------------------------%
\subsection{On the multi-particle MSA}       \label{ssec:on.MP.MSA}

In the proof of Theorem \ref{thm:main} we will focus on properties of
finite-volume versions $\B{H}_{\boxl}=\B{H}_{\boxl}^{(N)}(\omega )$
of Hamiltonian
$\B{H}$. More precisely, let
$\boxl=\boxl^{(N)}(\B{u})$ be an $N$-particle box and consider the operator
$\B{H}^{(N)}_{\boxl}(\omega )$ in $L^2(\boxl)$,
(referred  to as the Hamiltonian of the $N$-particle system in $\boxl$)
of the same structure as in \eqref{eq:random.hamiltonian}, \eqref{eq:interaction}  and
\eqref{eq:external}:
\begin{equation}        \label{eq:box.hamiltonian}
\B{H}^{(N)}_{\boxl}(\omega )=
-\frac{1}{2}\B{\Delta}^{\boxl}+\B{U}+\B{V}(\omega).
\end{equation}

Here $\B{\Delta}^{\boxl}$ stands for the Laplacian
in $\boxl$ with Dirichlet's boundary conditions on $\partial\boxl$.

The spectrum of a given operator, e.g., $\B{H}$, will be denoted as $\sigma(\B{H})$.

Under conditions \IOne-\ITwo\, and \EOne-\EFour, for
for $\prob$-almost all $\omega\in\Omega$, operator
$\B{H}_{\boxl}^{(N)}(\omega )$
is correctly defined in $L_2({\boxl})$, as a unique self-adjoint extension
from the domain $C^2_0(\boxl)$.
Moreover, $\B{H}^{(N)}_{\boxl}(\omega)$ has a discrete spectrum, since
its resolvent
\begin{equation}\label{Gresolvent}
\Green^{\boxl}(E)=\big(\B{H}_{\boxl}-E\big)^{-1},\;\hbox{ for }E\in
\D{R}\setminus\spec\big(\B{H}_{\boxl}\big),
\end{equation}
is a compact integral operator; it will be in the centre of our attention. Its  kernel \begin{equation}\label{Gfunction}
\boxl\times\boxl\ni(\B{x},\B{x}')\mapsto\green^{\boxl}(\B{x},\B{x}';E),\;\;
\B{x},\B{x}'\in{\boxl},
\end{equation}
is known as the Green function of $\B{H}_{\boxl}$.
The MSA is based on an asymptotical analysis of
resolvent $\Green^{\boxl}(E)$
as $\boxl\nearrow\D{R}^{Nd}$.
More precisely, boxes $\boxl$ will have the form
\[{\boxl}=\boxl_{L_k}(\B{u}),\;\B{u}\in\D{Z}^{Nd},\;
k=0,1,\ldots ,\]
where positive integers $L_k$ are determined by a recurrence
involving a starting value $L_0$ and a number $\alpha >1$:
\begin{equation}       \label{eq:length.scale}        %\eqno(1.16)
L_k= [ L_{k-1}^\alpha ] \sim (L_0)^{\alpha^k},\; k\geq 1.
\end{equation}
Here $[\cdot ]$ stands for the integer part. In future, we will
take $\alpha =3/2$.
Nevertheless, to keep a connection with the literary tradition, we will
continue using symbol $\alpha$. The same can be said about
parameter $\beta >0$ appearing in \eqref{eq:NR}: its value will
be $\beta =1/2$.

The positive integer value $L_0$ (the radius
of box $\boxl_{L_0}(\B{u})$) will be eventually assumed
to be large enough (depending on technical constants
emerging in the course of our argument, which, in turn, are
determined by $d$, $N$, $\{\varPhi^{(k)}\}$, $F$ and $\{\varphi_s\}$;
cf. \eqref{eq:interaction}, \eqref{eq:external}
and \eqref{eq:conditional.distribution}).
However, in several
definitions and related constructions the value of $L_0$ will
only have to satisfy some trivial restrictions, obvious
from the context.

To put it simply, one needs $L$ to be large enough for the asymptotic relations of the form
$\ln L \ll L^a  \ll e^{L^\beta} \ll e^{bL}$ (with $a,b>0,\beta\in(0,1)$) to hold.

Summarising, for future references,
\begin{equation} \label{eq:alphabetaLzero}
\alpha =\frac{3}{2}\,,\;\;\;
\beta =\frac{1}{2}\,,\;\;\;L_0\;\hbox{ is a  positive integer,
large enough.}
\end{equation}

To the reader familiar
with the MSA method in localisation proofs we can say at this
stage that the existence of values $\eta^*>0$, $m^*>0$ and $p^*>Nd$
claimed in Theorem \ref{thm:three} below will emerge as a result
of a `combined' induction, in the number of particles, $N$,
and the `scaling' index $k$ appearing in
in \eqref{eq:length.scale}.

Consequently, in the course of the argument, we will often
work with $n$-particle Hamiltonians
$\B{H}_{\boxl}^{(n)}(\omega )$, of the same form as in
\eqref{eq:box.hamiltonian}, with $n=1,\ldots ,N$.
{\it Mutatis mutandis},
definitions and facts introduced/noted for an $N$ particle system will
be used for a system of $n$ particles as well. (In fact,
some technical constructions will be carried on for an $n$-particle
system first, and then put in the context of $n$ running through the values
$n=1,\ldots ,N$).

Concluding this subsection, we stress that all
eigenvectors of finite-volume Hamiltonians appearing in
our arguments and calculations are normalised.

%---------------------------------------------------------------------%
%:s.1.7
%---------------------------------------------------------------------%
\subsection{Separable boxes and MSA estimates}\label{ssec:separable.boxes}

The principal difficulty encountered while attempting to extend
existing sigle-particle methods of localization theory to multi-particle
Anderson models arises from the (innocently looking) summatory formula
$\B{x}\in\D{R}^{Nd}\mapsto\sum\limits_{i=1}^NV(x_i;\omega )$
for the external potential in Eqn (\ref{eq:external.field}). In our context,
the values of the $N$-particle external potential exhibit, for
various points $\B{x}$, infinite-range correlations. For example,
suppose that vectors $\B{x}=(x_1,\ldots ,x_N)$ and $\B{x}'=(x'_1,\ldots ,
x'_N)$ include components $x_j$ and $x'_j$ with $\varphi_s(x_j-s)
\varphi_s(x'_j-s)\neq 0$ for some $s\in\D{Z}$ (which physically means that
the distance $|x_j-x'|$ is small). Then the random variable $\R{V}_s$
will be present in both sums $\sum\limits_{i=1}^NV(x_i;\omega )$
and $\sum\limits_{i=1}^NV(x'_i;\omega )$, generating their dependence
on each other.

This difficulty has been overcame in an analysis of regularity of
the so-called density of states (cf. \cite{KZ95}) and of eigenvalue distribution of finite-volume multi-particle Hamiltonians (cf. \cite{K08}).
Unfortunately, the information on regularity of the distribution
of eigenvalues in any given multi-particle box $\boxl_L(\B{x})$
does not provide a sufficient input
for the MSA. At the same time, the existence of the multi-particle
density of states is not required \textit{per se} for the MSA to work.

We tackle this issue by using the concept of
separability of boxes, which figures explicitly
in Theorems \ref{thm:two} and \ref{thm:three} below.
More precisely, it is the so-called Wegner-type bound
$\Wtwo{n}$ that requires the
notion of separability (see Eqn \eqref{eq:w-two}).
We want to stress here (as we did in \cite{BCSS010})
that the MSA requires Wegner-type bounds of two types: \textbf{(i)} for
one multi-particle box $\boxl$ and \textbf{(ii)} for two multi-particle
boxes $\boxl$ and $\boxl'$. See bounds $\Wone{n}$ and $\Wtwo{n}$
in Eqns \eqref{eq:w-one} and \eqref{eq:w-two} below.
However, the MSA is
less sensitive to optimality in these bounds (which may be important
for other areas in physics of disordered systems).
For the first time this notion
has been used, in the context of a two-particle
lattice Anderson model, in \cite{CS08} and \cite{CS09a}. The extension to
the $N$-particle lattice case was carried out in \cite{CS09b}.

We now turn to the formal aspect of separability.
Given an $n$-particle box $\boxl^{(n)}_L(\B{u})\subset\D{R}^{nd}$
and $j=1,\ldots ,n$, denote by $\varPi_j\boxl^{(n)}(\B{u})\subset\D{R}^d$
the projection of  $\boxl^{(n)}(\B{u})$ to the $j$th factor in
$\D{R}^{nd}$:
if $\boxl_L(\B{u})=\prod\limits_{i=1}^n\varLambda_L(u_i)$ then
$\varPi_j\boxl^{(n)}(\B{u})=\varLambda_L(u_i)$. Further, define the
`full projection' $\varPi\boxl^{(n)}(\B{u})$ of $\boxl^{(n)}(\B{u})$:
\[
\varPi\boxl^{(n)}(\B{u})\coloneq\bigcup_{j=1}^n\;\varPi_j
\boxl^{(n)}(\B{u})\subset\D{R}^d.\]

%-----------------------%

%-----------------------%
\begin{definition}          \label{def:separable.boxes}
Let $n=1,\ldots ,N$ and assume $\C{J}$ is a non-empty
subset in $\accol{1,\dots,n}$.
We say that a box $\boxl^{(n)}_L(\B{y})$ is $\C{J}$-separable from box
$\boxl^{(n)}_L(\B{x})$  if
\begin{equation}    \label{eq:separable.boxes}
\Bigl(\bigcup_{j\in\C{J}}\varPi_j\boxl^{(n)}_{L+R}(\B{y})\Bigr)\bigcap\;
\Bigl(\bigcup_{i\notin\C{J}}\varPi_i\boxl^{(n)}_{L+R}(\B{y})
\bigcup\varPi\boxl^{(n)}_{L+R}(\B{x})\Bigr)=\varnothing,
\end{equation}
where $R$ is the constant from condition \EThree.

Next, a pair of boxes $\boxl^{(n)}_L(\B{x})$, $\boxl^{(n)}_L(\B{y})$
is said to be {\it separable} if, for some non-empty set
$\C{J}\subset\{1,\dots,n\}$, $\dist\;\left(\boxl_L^{(n)}(\B{x}),\boxl_L^{(n)}(\B{y})\right)
>2N(L+R)$ and
\begin{enumerate}[\textbullet]
\item
either $\boxl^{(n)}_L(\B{y})$ is $\C{J}$-separable from $\boxl^{(n)}_L(\B{x})$,
\item
or $\boxl^{(n)}_L(\B{x})$ is $\C{J}$-separable from $\boxl^{(n)}_L(\B{y})$.
\end{enumerate}
\end{definition}
%-----------------------%@
\medskip

In physical terms: let box $\boxl^{(n)}_L(\B{x})$ be $\C{J}$-separable from
$\boxl^{(n)}_L(\B{y})$ and consider two quantum $n$-particle systems, in
$\boxl^{(n)}_L(\B{x})$ and $\boxl^{(n)}_L(\B{y})$ (i.e., with Hamiltonians
$\B{H}^{\boxl^{(n)}_L(\B{x})}$ and $\B{H}^{\boxl^{(n)}_L(\B{y})}$). Then
the first system contains a `detached' subsystem, formed
by particles with labels from $\C{J}$, with the following property.
$\forall$ $\B{u}=(u_1,\ldots ,u_n)\in\boxl^{(n)}_L(\B{x})$ and
$\B{v}=(v_1,\ldots ,v_n)\in\boxl^{(n)}_L(\B{y})$, the collection
of random variables $\R{V}_s$ from RF
$\C{V}$ contributing into the external potential sum
$\sum\limits_{j\in\C{J}}V(x_j;\omega )$
is disjoint from similarly defined collections, for sums
$\sum\limits_{j\not\in\C{J}}V(x_j;\omega )$ and
$\sum\limits_{1\leq j\leq n}V(x'_j;\omega )$.
This implies independence of sum $\sum\limits_{j\in\C{J}}V(x_j;\omega )$
and the pair of sums $\sum\limits_{j\not\in\C{J}}V(x_j;\omega )$ and
$\sum\limits_{1\leq j\leq n}V(x'_j;\omega )$ and
provides enough `randomness' to produce
satisfactory estimates.

%-----------------------%
\begin{lemma}       \label{lem:CondGeomSep}
Given $n\geq 2$, set $\kappa=\kappa (n) = n^n$.
For any $L>1$ and $n$-particle configuration $\B{x}\in\D{Z}^{nd}$,
there exists a collection of
$n$-particle boxes $\boxl_{L^{(l)}}(\B{x}^{(l)})$, $l=1,\ldots,
K(\B{x},n)$, with
$K(\B{x},n)\leq\kappa$ and $L^{(l)}\leq 2n(L+R)$, such that
if a vector $\B{y}\in\D{Z}^{nd}$ satisfies
\begin{equation}  \label{eq:CondGeomSep}
\B{y}\notin\bigcup_{\ell=1}^{K(\B{x},n)}\boxl_{L^{(l)}}(\B{x}^{(l)}),
\end{equation}
then boxes $\boxl^{(n)}_L(\B{x})$ and $\boxl^{(n)}_L(\B{y})$
with $\dist\;\left(\boxl^{(n)}_L(\B{x}),\boxl^{(n)}_L(\B{y})\right)>2N(L+R)$
are separable.

In particular,
a pair of boxes $\boxl^{(n)}_{L}(\B{x})$, $\boxl^{(n)}_{L}(\B{y})$
with $\dist\;\left(\boxl^{(n)}_L(\B{x}),\boxl^{(n)}_L(\B{y})\right)>2NL$
is separable if
$$
\boxl^{(n)}_{L+R}(\B{y}) \bigcap \boxl^{(n)}_{|\B{x}| + L+R}(\B{0})
= \varnothing.$$
\end{lemma}
%-----------------------%

For the proof of Lemma \ref{lem:CondGeomSep}, see Section
\ref{sec:appendix}.
%-----------------------%
\begin{corollary}       \label{cor:CondGeomSep}
Fix two integers, $n\geq 2$ and $L>1$, and let  $\kappa <\infty$ be the
number defined in Lemma {\rm{\ref{lem:CondGeomSep}}}. Set $B = 4n(L+R)+1$ and
consider an $n$-particle box
$\boxl_{L}(\B{x})$ and  $2\kappa +1$ disjoint concentric annular
sets $\BS{A}_1(\B{x})$, $\ldots$, $\BS{A}_{2\kappa+1}$ around
$\boxl_L(\B{x})$:
$$
\BS{A}_{j}(\B{x}) = \boxl_{L + jB}(\B{x}) \setminus \boxl_{L + (j-1)B}(\B{x}),
\;\; j=1, \ldots, 2\kappa +1.
$$
Then at least one of the annuli $\BS{A}_{2j-1}(\B{x})$, $1 \le j \le \kappa + 1$, 
contains no box $\boxl_L(\B{y})$ not separable
from $\boxl_L(\B{x})$.
\end{corollary}

%-----------------------%
\begin{proof}
Assume otherwise and consider $\kappa +1$ boxes $\boxl_L(\B{y}_j)
\subset \BS{A}_{2j-1}(\B{x})$, $j=1, \ldots, \kappa +1$, which
are not separable from $\boxl_L(\B{x})$. Since
$$
\dist( \boxl_L(\B{y}_j), \boxl_L(\B{y}_{j+1}) )
\ge \dist(\BS{A}_{j}(\B{x}), \BS{A}_{j+1}(\B{x})) - 2( L +R) > 4n(L+R),
$$
these $\kappa (n)+1$ boxes cannot
be enclosed in $\kappa (n)$ boxes of radius $2n(L+R)$, in contradiction
to the first assertion of Lemma \ref{lem:CondGeomSep}.
\end{proof}
%-----------------------%

We would like to stress that
\begin{itemize}
\item
the value $\kappa$  depends only upon the number of particles $n$;
%\item
%moreover, all boxes  $\boxl^{(n)}_{L'}(\B{y})$ (with a fixed radius $L'$) which are %\textit{non-separable} from
%%a given box
%$\boxl^{(n)}_{L}(\B{x})$ are located inside a finite box
%concentric with $\boxl^{(n)}_{L}(\B{x})$;
\item
in the case where boxes  $\boxl^{(n)}_{L}(\B{x})$ and
$\boxl^{(n)}_{L}(\B{0})$ are disjoint, it is always true that box
$\boxl^{(n)}_{L}(\B{x})$ is $\C{J}$-separable from
$\boxl^{(n)}_{L}(\B{x})$, for some  $\C{J}\subseteq\{1,\ldots , n\}$.
\end{itemize}

Define the outer layer (of width $2$) in a box $\boxl_L(\B{u})$ and
its lattice counterpart $\boxx_L(\B{u})$:
\begin{equation}     \label{eq:box.out}
\boxl^{\out}_L(\B{u})=\boxl_L(\B{u})\setminus\boxl_{L-2}(\B{u}),\;\;
\boxx^{\out}_L(\B{u})=\boxl^{\out}_L(\B{u})\cap\,\D{Z}^{nd},
\;\;\B{u}\in\D{Z}^{nd}.
\end{equation}

For given $m>0$ and $L\geq 1$, set :
\begin{equation}\label{gammamLn}
\gamma(m,L,n) \;(\;= \gamma_N(m,L,n))\;\;
= mL\left(1 + L^{-1/4} \right)^{N-n+1},
\; 1 \le n \le N.
\end{equation}

%-----------------------%
\begin{definition}[$(E,m)$-nonsingularity]         \label{def:nonsingular}
Let $E\in\D{R}$ and $m>0$. We say that box
$\boxl=\boxl^{(n)}_L(\B{u})\subset\D{R}^{nd}$, $1\le n \le N$, is
\emph{$(E,m)$-nonsingular} ($(E,m)$-$\nonsing$) if $E\in\D{R}\setminus
\spec(\B{H}_{\boxl})$ and for any
$\B{y}\in\boxx^{\out}_L(\B{u})$, the $L_2$-norm of the operator
$\one_{\cell(\B{u})}\Green^{\boxl}(E)
\one_{\cell(\B{y})}$ satisfies the bound
\begin{equation}     \label{eq:nonsingular}     %\eqno(2.3)
\|\one_{\cell(\B{u})}\Green^{\boxl}(E)
\one_{\cell(\B{y})}\|_{L_2(\D{R}^{Nd})}\leq\ee^{-\gamma(m,L,n)}.
\end{equation}

Otherwise, $\boxl$ is called \emph{$(E,m)$-singular} ($(E,m)$-$\sing$).

Similarly, a lattice box $\boxx_L(\B{u})=\boxl_L(\B{u})\cap\D{Z}^{nd}$
is called $(E,m)$-NS or $(E,m)$-S when the Euclidean box
$\boxl_L(\B{u})$ is $(E,m)$-NS or $(E,m)$-S, respectively.
\end{definition}
%-----------------------%

Consider the following property:

%-----------------------%
\begin{enumerate}[$\DSmkn$:]
\item
Given $m>0$, $k=0,1,\ldots$ and an interval $I\subseteq\D{R}$, for any
pair of separable
boxes $\boxl_{L_k}^{(n)}(\B{u})$, $\boxl_{L_k}^{(n)}(\B{v})$,
the probability
\begin{equation}  \label{eq:DSmkn}
\prob\bigl\{\;\forall\;\;E\in I,\;\boxl_{L_k}^{(n)}
(\B{u})\text{ or }\boxl_{L_k}^{(n)}(\B{v})\text{ is }\;(E,m){\rm{-NS}}\;
\bigr\}\geq 1-L_k^{-2p}.
\end{equation}
\end{enumerate}
%-----------------------%

Recall: $L_k$ stands for an integer of the form
\eqref{eq:length.scale}, with $\alpha$ as in
\eqref{eq:alphabetaLzero}. The abbreviation DS means `double
singularity'.

%----- the next paragraph was moved here from the previous subsection -----%

Property $\DSmkN$ (with $n=N$), is critical for the $N$-particle
MSA scheme; see
Theorem \ref{thm:two} below. Once this property is established
for all $k\geq 0$ (at the end of Section~\ref{sec:caseIII}), it
will mark the end of the proof of Theorem~\ref{thm:main}.

%-----------------------%
\begin{theorem}              \label{thm:two}
Let $I\subseteq\D{R}$ be an interval. Assume that for
some $m>0$,  $L_0>2$, $p>Nd$  and for any $k\geq 0$, property
$\DSmkN$ holds true, with $L_k$ as in Eqns \eqref{eq:length.scale},
\eqref{eq:alphabetaLzero}.

Then, with $\prob$-probability one,

\begin{enumerate}[\rm(i)]
\item
The spectrum of $\B{H}^{(N)}(\omega)$ in $I$ is pure point.
\item
The eigenfunctions
$\BS{\varPsi}_j(\B{x};\omega)$ of Hamiltonian $\B{H}^{(N)}(\omega)$ with
eigenvalues $E_j(\omega)\in I$ satisfy the exponential bounds
similar to Eqn \eqref{eq:exponential.decay}:
\begin{equation}  \label{eq:DSk.decay}
\left\|{\mathbf 1}_{\B{C}(\B{u})}\BS{\varPsi}_j(\;\cdot\;;\omega)\right\|_{L_2
(\D{R}^{Nd})}
\leq c_j(\omega)\eul^{-m\abs{\B{u}}},\;\;\B{u}\in\D{Z}^{Nd}.
\end{equation}
\end{enumerate}
\end{theorem}
%-----------------------%

%-----------------------%
Theorem \ref{thm:two}
represents an `analytic' part of the MSA.
(Probability plays a subordinate role here, reduced merely
to the Borel-Cantelli lemma, which is guaranteed by the fact
that $p>Nd$.) The proof of Theorem \ref{thm:two}
is `standard', in the sense that it does not use particulars
of the model involved. We therefore omit the proof of Theorem
\ref{thm:two} from the paper, referring the reader to
\cite{CS08}*{Theorem 2} and \cite{CS09b}*{Theorem 2}. (In fact,
the proof of Theorem \ref{thm:two} follows almost {\it verbatim}
the proof of Theorem 2.3 from \cite{DK89}.)

%-----------------------%

In view of Theorem \ref{thm:two}, the assertion of Theorem
\ref{thm:main} can be deduced from the following Theorem \ref{thm:three}.

\begin{theorem}              \label{thm:three}
Under assumptions of Theorem \ref{thm:main}, there exist $\eta^*>0$ sufficiently small, 
$p^*>Nd$, and $m^*>0$ such that, for an integer $L_0>1$ large enough,
property $\DSmkN$ holds for all $k\geq 0$, with $p=p^*$, $m=m^*$,
$I=[E^0,E^0+\eta^*]$ and $L_k$ as in
Eqns \eqref{eq:length.scale}, \eqref{eq:alphabetaLzero}.
\end{theorem}
%-----------------------%

The rest of the paper is devoted to the proof of Theorem \ref{thm:three}.
This theorem represents a `probabilistic' part of the MSA;
unlike Theorem \ref{thm:two}, its
proof is quite sensitive to particulars of a given model. Nevertheless,
we will  follow the same logical scheme as in
\cite{CS09b}*{Theorem 3}.

%---------------------------------------------------------------------%
%:s.1.6
%---------------------------------------------------------------------%
\subsection{Comments on the structure of the paper}\label{ssec:structure}

\begin{itemize}
\item In Section 2, we adapt well-known "geometric
resolvent
inequalities", established for Schr\"{o}dinger operators in Euclidean
spaces. As a result, we state these inequalities in
a form convenient
for subsequent analysis of the above norm
$\|\one_{\cell(\B{u})}\Green^{\boxl_L(\B{v})}(E)
\one_{\cell(\B{w})}\|$.

\item In Section 3, following Ref. \cite{C08}, we discuss a useful notion
of `lattice subharmonicity'. It is subsequently used in Sections 6 and 7.
A reader familiar with the MSA may favour a different argument while
proving the main result of Section 3,
Lemma \ref{lem:RadialGF} (cf., e.g, the proof of Lemma 4.2 in \cite{DK89}) and skip the rest of Section 3.

\item Further, Section 4 describes the MSA inductive scheme
adopted in this paper and establishes the initial step of
the induction. We then discuss the structure of the argument
in the inductive step. To conduct the inductive step,
we have to analyse three types of pairs of separable
boxes $\boxl_{L_k}^{(n)}(\B{u})$, $\boxl_{L_k}^{(n)}(\B{v})$
figuring in Eqn \eqref{eq:DSmkn}. These types are described as
partially interactive, fully interactive or mixed pairs
$\boxl_{L_k}^{(n)}(\B{u})$, $\boxl_{L_k}^{(n)}(\B{v})$,
depending on a property of `decomposability' of the
corresponding particle systems into non-interacting subsystems.

\item In Sections 5--7 we give a case-by-case analysis of each
of the three aforementioned types.
Sections 4--7 are in fact adaptations, for multi-particle
alloy-type Anderson models, of the argument from Sections 4--7
of paper \cite{CS09b}
where the focus was on multi-particle lattice Anderson models.
Here we systematically refer to various results and techniques
for Schr\"odinger operators in a Euclidean space,
summarised and developed in the monograph \cite{St01}.

\item Section 8 is an appendix containing (elementary)
proofs of two basic (but convenient) facts used in
the main body of the paper.
\end{itemize}

%---------------------------------------------------------------------%
%:sec.2
%---------------------------------------------------------------------%
\section{Resolvent inequalities}          \label{sec:GRI}

%-----------------------%
Throughout Sections 2-3, we work with a fixed bounded
interval $I\subset\D{R}$ and variable $n = 1, ..., N$.

%---------------------------------------------------------------------%
%:s.2.1
%---------------------------------------------------------------------%
\subsection{Geometric resolvent inequality}

Given an $n$-particle box $\boxl_L(\B{u})\subset\D{R}^{nd}$
with $L\geq 4$, we define the interior $\boxl^{\intr}_L(\B{u})$ of
$\boxl_L(\B{u})$ by
\begin{equation}     \label{eq:box.out.in}       %\eqno(2.2)
\boxl^{\intr}_L(\B{u})=\boxl_{\croch{L/3}}(\B{u}).
\end{equation}
\def\wh{\widetilde}
Next, consider two $n$-particle boxes, $\boxl_L(\B{u})\subset
\boxl_{\wh L}(\B{u})$, with $4\leq L<{\wh L}$, and cellular subsets
\[ \D{A}\subset\boxl^{\intr}_L(\B{u})\;\hbox{ and }\;
\D{B}\subset{\boxl}_{\wh L}(\B{u})\setminus\boxl_L(\B{u}).
\]

From now on we will omit subscript $L_2(\D{R}^{nd})$
in the notation $\|\;\cdot\;\|_{L_2(\D{R}^{nd})}$
for the vector and operator norms in $L_2(\D{R}^{nd})$.
The standard resolvent identity for Schr\"odinger operators
combined with commutator estimates implies the following fact
(cf. \cite{St01}*{Lemma 2.5.2}):
\begin{enumerate}[\textbf{(GRI)}]
\item
\emph{Geometric Resolvent Inequality}:\\
Let $\boxl_L(\B{u})$, $\boxl_{\wh L}(\B{u})$,
$\D{A}$ and $\D{B}$ be as above.
Then, $\forall$ $E\in I\setminus\Big(\spec\left(\B{H}_{\boxl_L(\B{u})}\right)
\cup\spec\left(\B{H}_{{\boxl}_{\wh L}(\B{u})}\right)\Big)$, the operator norms
satisfy
\begin{equation}         \label{eq:GRI}     %\eqno(2.3)
\norm{\one_{\D{B}}\Green^{{\boxl}_{\wh L}(\B{u})}(E)\one_{\D{A}}}
\leq C^{(0)}\norm{\one_{\D{B}}
\Green^{{\boxl}_{\wh L}(\B{u})}(E)\one_{\boxl^{\out}_L(\B{u})}}
\times\norm{\one_{\boxl^{\out}}\Green^{{\boxl}_L(\B{u})}(E)\one_{\D{A}}}.
\end{equation}
\end{enumerate}
Here $C^{(0)}>0$ is a `geometric' constant: owing to the condition
$4\leq L<\tilde L$, this constant depends only on $n$ (and is
uniformly bounded for $1\leq n\leq N$), but not on $E$.
See~\cite{St01}*{Lemma 2.5.4}. Later in this section, some other positive constants
will appear, of a similar nature; we will denote them by $C^{(1)}$, $C^{(2)}$ and
so on.

%-----------------------%

%---------------------------------------------------------------------%
%:s.2.2
%---------------------------------------------------------------------%
\subsection{Discretized Green functions}

Inequality \eqref{eq:GRI} will enable us to use
the function
\[
\boxx_L(\B{u})\times\boxx_L(\B{u})\ni(\B{v},\B{y})
\mapsto \|\one_{\cell(\B{v})}\Green^{\boxl_L(\B{u})}(E)
\one_{\cell(\B{y})}\|
\]
figuring in \eqref{eq:nonsingular} as a discretization
of original Green functions $\green^{{\boxl}_L(\B{u})}
(\B{x},\B{x}';E)$. Consequently, we will be able to apply
a number of technical
arguments developed earlier for multi-particle lattice Anderson
models; see \cite{CS08}, \cite{CS09a}, \cite{CS09b}.

Let ${L}>7$ and consider boxes $\boxl_{L}(\B{u}\,)$ and
$\boxx_{L}(\B{u}\,)$. Further, pick a point
$\B{v}\in{\boxx}_{L}(\B{u}\,)$ and integer $\ell$, with
$3<\ell< L-3$, such
that $\boxl_\ell(\B{v})\subset\boxl_{{L}-3}(\B{v})$. As above
(see~\eqref{eq:box.out}), set
\begin{subequations}      \label{eq:boxes.out}
\begin{equation}   \label{eq:continuous.boxes.out}        %\eqno(2.4.1)
{\boxl}^{\out}_{L}(\B{u}\,)=\boxl_{L}(\B{u}\,)
\setminus\boxl_{\tilde L-2}(\B{u}\,)\;\hbox{and}\;
\boxl^{\out}_\ell(\B{v})=\boxl_\ell(\B{v})\setminus\boxl_{L-2}(\B{v}),
\end{equation}
and
\begin{equation}         \label{eq:lattice.boxes.out}        %\eqno(2.4.2)
{\boxx}^{\out}_{L}(\B{u}\,)={\boxl}^{\out}_{L}
(\B{u}\,)\cap\,\D{Z}^{nd},
\;\hbox{and}\;
\boxx^{\out}_\ell(\B{v})=\boxl^{\out}_\ell(\B{v})\cap\,\D{Z}^{nd}.
\end{equation}
\end{subequations}
We have, evidently,
\[
{\boxl}^{\out}_{L}(\B{u}\,)\subset
\bigcup_{\B{w}\in{\boxx}^{\out}_{L}(\B{u}\,)}\B{C}(\B{w})
\;\hbox{ and }\;
\boxl^{\out}_\ell(\B{v})\subset\bigcup_{\B{w}\in\boxx^{\out}_\ell(\B{v})}
\B{C}(\B{w}).
\]
Hence, for any $\B{x}\in\D{R}^{nd}$, the indicator functions
obey
\begin{equation}  \label{eq:}        %\eqno(2.5)
\one_{{\boxl}^{\out}_{L}(\B{u}\,)}(\B{x})
\leq\sum_{\B{w}\in{\boxx}^{\out}_{L}(\B{u}\,)}
\one_{\B{C}(\B{w})}\;(\B{x})\;\hbox{ and }\;
\one_{\boxl^{\out}_\ell(\B{v})}\;(\B{x})\leq
\sum_{\B{w}\in\boxx^{\out}_\ell(\B{v})}\one_{\B{C}(\B{w})}(\B{x}).
\end{equation}
Now the  Eqn \eqref{eq:GRI}
implies, for $\B{y}\in{\boxx}^{\out}_{L}(\B{u}\,)$
and $E\in I\setminus\Big(\spec(\B{H}_{{\boxl}_\ell(\B{v})})\cup
\spec(\B{H}_{{\boxl}_{L}(\B{u}\,)}\Big)$,
\begin{equation}  \label{eq:dgri-one}        %\eqno(2.6.1)
\begin{array}{l}
\norm{\one_{\B{C}(\B{v})}\Green^{{\boxl}_{L}(\B{u}\,)}(E)
\one_{\B{C}(\B{y})}}\\ \\
\qquad\leq C^{(0)}\sum\limits_{\B{w}\in\boxx^{\out}_\ell(\B{v})}
\norm{\one_{\B{C}(\B{v})}\Green^{{\boxl}_\ell(\B{v})}(E)
\one_{\B{C}(\B{w})}}\times\norm{\one_{\B{C}(\B{w})}
\Green^{{\boxl}_{L}(\B{u}\,)}(E)\one_{\B{C}(\B{y})}}.
\end{array}\end{equation}

%-----------------------%
\begin{definition}[Discretized Green function]
\label{def:discr.green}

Given boxes $\boxl=\boxl_L(\B{u})$ and $\boxx=\boxx_L(\B{u})$, value
$E\in\D{R}\setminus\spec(\B{H}_{\boxl})$ and vectors $\B{v},\,\B{w}
\in\boxx$, we now denote
\begin{equation}  \label{eq:discr.green}        %\eqno(2.7)
\R{D}_{L,\B{u}}(\B{v},\B{w};E)=\norm{\one_{\cell(\B{v})}
\Green^{\boxl}(E)\one_{\cell(\B{w})}}.
\end{equation}
We call function $\boxx\times\boxx\ni(\B{v},\B{w})\mapsto
\R{D}_{L,\B{u}}(\B{v},\B{w};E)$
the \emph{discretized Green function} for $\B{H}_{\boxl}$.
The same definition is applicable for
$\boxl_\ell (\B{u})$ and $\boxx_\ell (\B{u})$ yielding
function $\boxx_\ell\times\boxx_\ell\ni(\B{v},\B{w})\mapsto
\R{D}_{\ell,\B{u}}(\B{v},\B{w};E)$.
\end{definition}
%-----------------------%

It is worth to keep in mind that
$\R{D}_{L,\B{u}}(\B{v},\B{w};E)=\R{D}_{L,\B{u}}(\B{w},\B{v};E)\geq 0$,
$\B{v},\B{w}\in\boxx_L(\B{u})$.

The bound in Eqn \eqref{eq:dgri-one} now takes the following form:

%-----------------------%
\begin{enumerate}[\textbf{(DGRI)}]
\item
\emph{Discretized geometric resolvent inequality}:
Given boxes
$\boxl_\ell(\B{v})\subset\boxl_{L-3}(\B{u})$,
$\forall$ $\B{y}\in\boxx^{\out}_{L}(\B{u})$ and
$E\in I\setminus\Big(\spec(\B{H}_{{\boxl}_\ell(\B{u})})\cup
\spec(\B{H}_{{\boxl}_{L}(\B{u}\,)})\Big)$,
\begin{equation}  \label{eq:dgri-two}
\R{D}_{L,\,\B{u}}(\B{v},\B{y};E)\leq
C^{(0)}\sum_{\B{w}\in\boxx_\ell^{\out}(\B{v})} \,
\R{D}_{\ell ,\B{v}}(\B{v},\B{w};E)    \,   \R{D}_{L,\B{u}}(\B{w},\B{y};E).
\end{equation}
\end{enumerate}
%-----------------------%

Our task in the remaining part of the paper will be
essentially reduced to the analysis of decay  of
functions $\R{D}_{L_k,\B{u}}(\B{v},\B{w};E)$ for $E\in\D{R}\setminus
\spec(\B{H}_{\boxl_{L_k}(\B{u})})$,
when vectors $\B{v}$ and $\B{w}$ are distant apart
(viz., $\B{v}$ is `deeply' inside $\boxx_{L_k}(\B{u})$
whereas $\B{w}$ is near the boundary of $\boxx_{L_k}(\B{u})$; see below).

Working with  a lattice box
$\boxx=\boxx_L(\B{u})\subset\D{Z}^{nd}$, we will use
the inner boundary $\partial^-\boxx$:
\begin{equation}
\begin{split}
\partial^-\boxx
&\coloneq\accol{\B{x}\in\boxx:\,
\dist(\B{x},\D{Z}^{nd}\setminus\boxx)=1}, \\
\end{split}
\end{equation}
Similar notion can be introduced also for a general cellular set $\D{A}$.

One of the key points in the proof of Theorem \ref{thm:three}
is an exponential upper bound on discretized Green functions
in finite boxes (cf.~Eqn \eqref{eq:RadialGF} in Lemma \ref{lem:RadialGF} below).
This bound is obtained with the help of Lemma \ref{lem:Radial} using the notion of ``lattice subharmonicity'' introduced in the next section.

%---------------------------------------------------------------------%
%:sec.3
%---------------------------------------------------------------------%
\section{The scaling step inquality}     \label{sec:subharmonicity}

%---------------------------------------------------------------------%
%:s.3.1
%---------------------------------------------------------------------%

The main result of this section is the bound \eqref{eq:RadialGF} established in  Lemma \ref{lem:RadialGF} below.
It is in fact based on a construction alternative to
\cite{DK89}*{Lemma 4.2, Section 4} but serving the same purpose.
A similar construction was used earlier in \cite{CS09b}, in the
framework of a multi-particle tight-binding
Anderson model.

\begin{lemma}\label{lem:RadialGF}
Given $n=1,\ldots ,N$, $m>0$, and a positive integer $K$, consider an
$n$-particle box
$\boxl_L(\B{u})$. There exists a value
$L^*_{\rm{sc}} =L^*_{\rm{sc}}(m,K)$ with the following property. Suppose that the conditions
{\rm{(A)-(C)}} are satisfied:
\begin{enumerate}[{\rm (A)}]
\item $ L\geq L^*_{\rm{sc}}$.
\item $\boxl_L(\B{u})$ is $E$-{\rm{CNR}}.
\item there exists a (possibly empty) family $\B{A} = \{\B{A}_i, 1 \le i \le J\}$
of disjoint annuli $\B{A}_i = \boxl_{l_i+r_i}(u)\setminus \boxl_{l_i}(u)$ of total
width $r_1 + \cdots + r_J \leq KL^{1/\alpha}$ such that any box
$\boxl_\ell(\B{v})\subset \boxl_L(\B{u}) \setminus \B{A}$ is {\rm{NS}}.
\end{enumerate}
Then box $\boxl_L(\B{u})$ is {\rm{NS}}:
\begin{equation} \label{eq:RadialGF}
\max_{\B{y}\in\partial^-\!\boxx_L(\B{u})}\abs{\green^{\boxl_L(\B{u})}
(\B{u},\B{y};E)}\leq\ee^{-\gamma(m,L,n)}.
\end{equation}
\end{lemma}

The proof of Lemma \ref{lem:RadialGF} is completed at the end of the section; it is
based on a number of auxiliary statements which occupy the rest of Section
\ref{sec:subharmonicity}.
Lemma \ref{lem:RadialGF} will be used in Section \ref{sec:caseII} with $n=N$ and
$K=\kappa (N)$, where
$\kappa (N)=N^n$ is the constant from Lemma \ref{lem:CondGeomSep}.

\subsection{DGRI for NS boxes}    \label{ssec:DGRI-NS}

Suppose
that a number $m>0$ has been given,
and consider an arbitrary point $E$ from the bounded interval $I$. Consequently, we refer to $(E,m)$-NS and
$(E,m)$-S boxes as NS- and S-boxes, assuming that $E$
does not lie in the spectra of the corresponding operators.

The aim is is to derive, from the Eqn \eqref{eq:dgri-two},
an effective procedure of estimating the decay of the discretized
Green functions $\R{D}_{L,\B{u}}(\B{v},\B{w};E)$ when vectors
$\B{v}$ and $\B{w}$ are far from each other.

Given a positive integer $\ell <L$, assume that box
$\boxl_L(\B{u})$ does not contain an S-box
$\boxx_{\ell}(\B{v})$. Then
Eqn \eqref{eq:dgri-two} implies that for any
site $\B{y}\in\partial^-\boxx_L(\B{u})$ and any box
$\boxl_{\ell}(\B{v})\subset\boxl_L(\B{u})$:
\begin{equation}     \label{eq:DGRI.nonsingular}         %\eqno(3.1)
0\leq\R{D}_{L,\B{u}}(\B{v},\B{y};E)
\leq\R{b}_1\max_{\substack{ \B{w}\in\boxx_L(\B{u}) \\  {\abs{\B{w}-\B{v}}= \ell }  } }
\;\R{D}_{L,\B{u}}(\B{w},\B{y};E)
\end{equation}
Here
\[
\R{b}_1=C^{(1)}\ee^{-m\ell}\ell^{Nd-1},
\]
and $C^{(1)}=C^{(1)}(N)$ is another `geometric' constant.

%---------------------------------------------------------------------%
%:s.3.2

%-----------------------%
\begin{definition}[$E$-complete non-resonance]    \label{def:NR}
Set $\beta=1/2$, $\alpha=3/2$ (cf.1.24). Given $E\in I$ and
$\B{v}\in\D{Z}^{nd}$, the $n$-particle box $\boxl_L(\B{v})$
and the corresponding lattice box $\boxx_L(\B{v})$ are called
\begin{enumerate}[(i)]
\item
\emph{$E$-nonresonant} ($E$-NR) if
\begin{equation} \label{eq:NR}
\dist\bigl(E,\spec(\B{H}_{\boxl_L(\B{v})})\bigr)\geq \ee^{-L^{\beta}},
\end{equation}
and $E$-resonant ($E$-R) if the opposite inequality holds;
\item
\emph{$E$-completely non-resonant} ($E$-CNR) if
$\boxl_L(\B{v})$ is $E$-NR and does not contain any
$E$-resonant box $\boxl^{(n)}_{\ell}(\B{w})$ with
$\ell\geq L^{1/\alpha}$.
\end{enumerate}
\end{definition}
%-----------------------%

%---------------------------------------------------------------------%
\subsection{DGRI for non-resonant S-boxes}  \label{ssec:DGRI-S}

Next, consider a situation where the box
$\boxl_L(\B{u})$ contains an $(E,m)$-S
box $\boxl_{\ell}(\B{v})$. Here $E\in I$, $m>0$, $1\leq \ell <L$
and $\B{v}\in\boxx_L(\B{u})$.

Suppose that
\begin{enumerate}[(i)]
\item
any box $\boxl_{\ell}(\B{w})\subset\boxl_L(\B{u})$,
with $\B{w}\in\boxx_L(\B{u})$ such that
$\dist(\boxl_{\ell}(\B{v}),\boxl_{\ell}(\B{w}))=1$, i.e.,
$\abs{\B{v}-\B{w}}=2\ell+1$, is NS;
\item
all boxes $\boxl_{\ell'}(\B{v}')\subset\boxl_L(\B{u})$
with $\B{v}'\in\boxx_L(\B{u})$ and $\ell\leq\ell'\leq L$ are $E$-NR.
\end{enumerate}
In this situation, Eqn \eqref{eq:dgri-two} implies that,
$\forall$ $\B{y}\in\D{Z}^{nd}\cap\partial^-\boxl_L(\B{u})$
and $\forall$ box $\boxl_{\ell}(\B{v})\subset\boxl_L(\B{u})$,
\begin{equation}        \label{eq:first.bound}     %\eqno(3.3)
\begin{array}{r}
\R{D}_{L,\B{u}}(\B{v},\B{y};E)\leq C^{(0)}\ee^{\ell^\beta}
\abs{\partial^-\boxl_{\ell}(\B{v})}
\max\;\Big[\R{D}_{L,\B{u}}(\B{w},\B{y};E):\qquad{}\\
\qquad
\B{w}\in\boxx_L(\B{u}),
\boxl_{\ell}(\B{w})\subset\boxl_L(\B{u}),
      \abs{\B{w}-\B{v}}=2\ell+1\big].\end{array}
\end{equation}
Applying Eqn \eqref{eq:dgri-two} to all
boxes $\boxl_{\ell}(\B{w})\subset\boxl_L(\B{u})$ neighboring
$\boxl_{\ell}(\B{w})$, we get the bound
\begin{equation}       \label{eq:second.bound}            %\eqno(3.4)
\begin{array}{l}\R{D}_{L,\B{u}}(\B{v},\B{y};E)\leq\R{b}_2
\max\;\Big[\R{D}_{L,\B{u}}(\B{w},\B{y};E):\\
\qquad\qquad \B{w}\in\boxx_L(\B{u}),
\boxl_{\ell}(\B{w})\subset\boxl_L(\B{u}),
      \abs{\B{w}-\B{v}}=2\ell+1\Big],\end{array}
\end{equation}
with
\begin{equation}       \label{eq:second.constant}            %\eqno(3.5)
\R{b}_2=C^{(2)}\ee^{-m\ell}\ee^{\ell^\beta}\ell^{d-1},
\end{equation}
where $C^{(0)}>0$ is yet another `geometric'' constant.

Observe also that $\R{b}_1\leq\R{b}_2$, so that
\eqref{eq:DGRI.nonsingular} implies a weaker inequality
\begin{equation}     \label{eq:weak.DGRI}     %\eqno(3.7)
\R{D}_{L,\B{u}}(\B{u},\B{y};E)
\leq\R{b}_2\max_{\B{v}\in\partial^-\boxl_{\ell}(\B{u})}
\R{D}_{L,\B{u}}(\B{v},\B{y};E).
\end{equation}

We see that the difference between the bounds
\eqref{eq:DGRI.nonsingular} and \eqref{eq:second.bound} resides
primarily in the form (and size) of the `reference set' of
points $\B{w}$ under the sign of $\max$.

%---------------------------------------------------------------------%
%:s.3.3
%---------------------------------------------------------------------%
\subsection{Multiple singular boxes}
\label{ssec:clustDSJ}

%---------------------------------------------------------------------%
\subsubsection{Singular chain}
$\,$

Given a positive integer $\ell <L$ and an energy $E\in I$,
suppose that $\boxl_L(\B{u})$ contains
some S-boxes of radius $\ell$ with centers in $\boxx_L(\B{u})$.
In order to be able to apply \eqref{eq:second.bound} to a given
S-box $\boxl_{\ell}(\B{v}^{(1)})\subset\boxl_L(\B{u})$,
$\B{v}^{(1)}\in\boxx_L(\B{u})$, we would need all boxes of radius
$\ell$ which neighbor $\boxl_{\ell}(\B{v}^{(1)})$, lie in
$\boxl_L(\B{u})$ and are  centered at a point in
$\boxx_L(\B{u})$ to be NS. However, one or more of these neighbors,
say $\boxl_{\ell}(\B{v}^{(2)})$, can be S. In such a case we pass to a larger box, $\boxl_{2\ell}(\B{v}^{(1)})\supset\boxl_{\ell}(\B{v}^{(1)})$,
and check for non-singularity of its neighbors
$\boxl_{\ell}(\B{v}^{(3)})\subset\boxl_L(\B{u})
\setminus\boxl_{2\ell}(\B{v}^{(1)})$, with $\dist
[\boxl_{2\ell}(\B{v}^{(1)}), \boxl_{\ell}(\B{v}^{(3)})]=1$. Again,
at least one of these boxes can be S. Then we pass to a larger
box $\boxl_{3\ell}(\B{v}^{(1)})$ and repeat the procedure.
In the end we obtain a finite sequence of S-boxes
\[
\boxl_{\ell}(\B{v}^{(1)}),\dots,\boxl_{\ell}(\B{v}^{(s)})
\subset\boxl_L(\B{u}),\;\hbox{ where } s\geq 1,
\]
with
\[
\dist\bigl(\boxl_{(t-1)\ell-1}(\B{v}^{(1)}),
\boxl_{\ell}(\B{v}^{(t)})\bigr)=1,\quad 2\leq 2\leq s,
\;\hbox{when $s\geq 2$.}
\]
We call such a sequence a singular chain, or, briefly, an S-chain,
of length $s$.

It is not hard to see that if $\boxl_L(u)$ contains no S-chain of length
$\geq K$, then for any point $\B{y}\in\boxx_{L-2K\ell}(\B{u})$
(i.e., not too close to the boundary of
$\boxx_L(\B{u})$) the following inequality holds true:
\begin{equation}         \label{eq:}  %\eqno(3.9)
\begin{array}{l}
\R{D}_{L,\B{u}}(\B{v},\B{y};E)\\
\;\;\leq Q\,\max\;\big[\R{D}_{L,\B{u}}(\B{w},\B{y};E):\;
\B{w}\in\boxx_L(\B{u}),
\abs{\B{w}-\B{v}}=(A+1)\ell-1\big]. \end{array}
\end{equation}
Here $A=A(\B{v})\leq 2K$, and the factor $Q>0$ is
assessed by
\begin{equation}
Q\leq C^{(3)}(2(A+1)\ell+1)^{nd-1}e^{-\gamma(m,\ell,n)}.
\end{equation}

%---------------------------------------------------------------------%
\subsubsection{Singular chains and separability}

Let $\kappa (n)$ be the value from Lemma \ref{lem:CondGeomSep}.
By  Corollary \ref{cor:CondGeomSep}, if we take
$\kappa (n)+1$ disjoint annuli of
width
\begin{equation}\label{B(n.ell)}
B = B(n,\ell) = 2n\ell+1
\end{equation}
with centre at $\B{v}$:
$$
\BS{A}_{j}(\B{v})
= \boxl_{\ell + 2jB}(\B{v}) \setminus \boxl_{\ell + (2j-1)B}(\B{v}), \quad
1\leq j\leq \kappa (n)+1,
$$
then at least one of them  contains no box $\boxl_L(\B{y})$ \textbf{not} separable
from $\boxl_{\ell}(\B{v})$.

%-----------------------%
\begin{definition}[A bad box] \label{def:bad.box}
An $n$-particle box $\boxl_\ell(\B{v})$  is called $(E,m)$-bad if it satisfies
the following conditions:
\begin{itemize}
\item $\boxl_\ell(\B{v})$ is $(E,m)$-singular;
\item each annulus $\BS{A}_{j}(\B{v})$, $1\le j \le \kappa (n)+1$,
contains an $(E,m)$-singular box $\boxl_\ell(\B{w}_j)$.
\end{itemize}
\end{definition}
%-----------------------%

The meaning of Definition \ref{def:bad.box} is that at least one of
the $(E,m)$-singular boxes  $\boxl_\ell(\B{w})$ must be separable from
$\boxl_\ell(\B{v})$.

%-----------------------%
\begin{definition}[The enveloping box] \label{def:boxed.singular.cluster}
Consider a finite, non-empty S-chain originating at $\B{v}$
and assume that $\boxl_\ell(\B{v})$ is not $(E,m)$-bad.
The \emph{enveloping box} for this S-chain  associated is
the smallest box $\boxl_{\widetilde L}(\B{v})$
centered at $\B{v}$ and containing this $S$-chain.
\end{definition}
%-----------------------%

By construction of an enveloping box $\boxl_{\widetilde L}(\B{v})$,
any box of radius $\ell$
adjacent to its boundary $\partial\boxl_{\widetilde L}(\B{v})$ must
be NS. When we restrict
ourselves to box $\boxl_L(\B{u})$,
we should always check if $\boxl_{\widetilde L}(\B{v})
\subset\boxl_{L-2\ell -1}(\B{u})$, i.e., box $\boxl_{\widetilde L}(\B{v})$
lies at distance $\ge (2\ell+1)$
from the boundary $\partial^- \boxl_L(\B{u})$, so that every
box of radius $\ell$ neighboring $\boxl_{\widetilde L}(\B{v})$
fits in $\boxl_L(\B{u})$.

We finish subsection \ref{ssec:clustDSJ} with the following
result concerning enveloping boxes.

\begin{lemma}\label{lem:sep.BS}
Suppose that box $\boxl_L(\B{u})$ contains no separable
pair of $(E,m)$-singular boxes of radius $\ell$. Then, for any point
$\B{v}\in\boxl_{L - (2\kappa (N)+1)\ell}(\B{u})$,
one of the following alternatives occurs:
\begin{enumerate}
\item $\boxl_{\ell}(\B{v})$ is $(E,m)$-non-singular.
\item There exists a box $\boxl_{\widetilde L}(\B{v})$
of radius ${\widetilde L}\leq (2\kappa (N)+1)\ell$ such that any box of radius $\ell$ adjacent to the boundary
$\partial\boxl_{\widetilde L}(\B{v})$ is $(E,m)$-non-singular.
\end{enumerate}
\end{lemma}

\proof In view of Corollary \ref{cor:CondGeomSep}, given a point
$\B{v}\in\boxl_{L - (2\kappa (N)+1)\ell}(\B{u})$,
and a collection of at least $\kappa (n)+1$ disjoint
annuli $\BS{A}_{j}(\B{v}) = \boxl_{\ell + jR}(\B{v}) \setminus
\boxl_{\ell + (j-1)B}(\B{v})$,
at least one of these annuli contains
only $\ell$-boxes separable with $\boxl_\ell(\B{v})$.
\qedhere

%---------------------------------------------------------------------%
%:s.3.4
%---------------------------------------------------------------------%
\subsection{Subharmonic functions in $\boxx_L(\B{u})$}
\label{ssec:Subharm.boxes}

%---------------------------------------------------------------------%
\subsubsection{Formal definition}

%-----------------------%
\begin{definition}[Subharmonicity]          \label{def:SH}
Fix constants $Q>0$, $A>1$ and integers $1 < \ell<L$, and let
$\B{S}$ be a subset in $\boxx_L(u)$. A nonnegative function $
f\colon\boxx_L(\B{u})\to\D{R}_+$
is called \emph{$(Q,A,\ell,\B{S})$-subharmonic} if it satisfies the
following properties:
\begin{enumerate}[(i)]
\item
$\forall$ point $\B{x}\in\boxx_L(\B{u})\setminus\B{S}$ with
$\dist\bigl(\B{x},\partial^-\!\boxx_L(\B{u})\bigr)\geq\ell$ we have
\begin{equation}   \label{eq:estim.f.out.singular.clusters}   %\eqno(2.15)
f(\B{x})\leq Q\,\max\;\big[f(\B{w}):\;\B{w}\in\boxx_L(\B{u}),
\abs{\B{w}-\B{x}}=2\ell+1\big].
\end{equation}
\item
$\forall$ point $\B{x}\in\B{S}$ $\exists$ an integer
$\rho(\B{x})$ $\ell\leq \rho (\B{x})\leq A\ell$, such that
\begin{equation}   \label{eq:estim.f.in.singular.clusters}  %\eqno(2.16)
\begin{array}{r}
f(\B{x})\leq Q\,\max\;\Big[f(\B{w}):\;
\B{w}\in\boxx_L(\B{u}),\qquad\\
\rho(\B{x})\leq\abs{\B{w}-\B{x}}\leq\rho(\B{x})+2\ell+1\Big]
\end{array}
\end{equation}
\end{enumerate}
\end{definition}
%-----------------------%

%-----------------------%
Next, following \cite{C08}*{Lemma 4.3}, we give a general
bound for subharmonic functions.

%-----------------------%
\begin{lemma}  \label{lem:Radial}
Suppose that a function $f\colon\boxx_L(\B{u})\to\D{R}_+$ is
$(Q,A,\ell,\B{S})$-subharmonic, and that $\B{S}$ can be covered by a collection of annuli with centre $\B{u}$ of total width $W=W(\B{S})$. Then
\begin{equation}    \label{eq:radial.descent}     %\eqno(3.19)
f(\B{u})\leq Q^{(L-W-3\ell)/\ell}
\max_{\B{x}\in\boxx_L(\B{u})} f(\B{x}).
\end{equation}

\end{lemma}
%-----------------------%

%-----------------------%
\begin{proof}
See \cite{C08}*{Proof of Lemma 4.3}.

\end{proof}
%-----------------------%

%---------------------------------------------------------------------%
%-----------------------%
\begin{lemma} \label{Lem_subarmonic}
Consider a lattice box  $\boxx_L(\B{u})\subset \D{Z}^{nd}$
and suppose that the following assumptions are satisfied:
\begin{enumerate}
  \item $\boxx_L(\B{u})$ is $E$-CNR;
  \item $\boxx_L(\B{u})$ contains no $(E,m)$-bad box;
  \item all $(E,m)$-{\rm S} boxes of radius $\ell$ inside $\boxx_L(\B{u})$
  can be covered by a set $\B{S}$.
\end{enumerate}
Then the function
\begin{equation}   \label{eq:subharmonic.function}     %\eqno(3.15)
f(\B{x}):=\max_{\B{y}\in\partial^-\!\boxx_L(\B{u})}
\R{D}_{L,\B{u}}(\B{x},\B{y};E)
\end{equation}
is $(Q,\ell,\B{S})$-subharmonic with
\begin{equation}     \label{eq:subharmonic.constant}      %\eqno(3.16)
Q=C^{(4)}(d)(n\ell)^{d-1}\ee^{\ell^{\beta}} \ee^{-\gamma(m,\ell,n)}.
\end{equation}
\end{lemma}
%-----------------------%
\proof A straightforward application of Lemma \ref{lem:Radial}.

Now suppose that any family of disjoint S-boxes
\[
\boxx_\ell(\B{v}^{(1)}),\;\boxx_\ell(\B{v}^{(2)}),\dots,
\boxx_\ell(\B{v}^{(j)}) \subset \boxx_L(\B{u}) \subset \D{Z}^{nd}
\]
corresponding to the cubes
\[
\boxl_\ell(\B{v}^{(1)}),\;\boxl_\ell(\B{v}^{(2)}),\dots,
\boxl_\ell(\B{v}^{(j)}) \subset \boxl_L(\B{u})  \subset \D{R}^{nd}
\]
contains at most $J$ elements, for some fixed $J<\infty$.
Then the function $f$ defined in \eqref{eq:subharmonic.function}
is $(Q,\ell,\B{S})$-subharmonic, with $Q$ as in
\eqref{eq:subharmonic.constant}, and with some set $\B{S}$
(in general, \emph{not} unique) can be covered by
a union $\B{A}(\B{S})$
of  concentric annuli $\B{A}_1$, $\ldots$, $\B{A}_j$:
\begin{equation}     \label{eq:annular.ring}      %\eqno(3.17)
\B{A}_i=\boxx_{l_i+r_i}(\B{u})\setminus\boxx_{l_i}(\B{u}),\quad 1\leq i\leq j.
\end{equation}
Here $0 < \l_1 < l_1 + r_1 < l_2 < \ldots < l_j+r_j < L$.

\subsection{Proof of Lemma \ref{lem:RadialGF}}\label{ssec:proof.lem.radialGF}

Owing to Lemma \ref{Lem_subarmonic}, it suffices to apply Lemma \ref{lem:Radial} to functions
$f:\,\B{v} \mapsto \R{D}_{L,\B{u}}(\B{v},\B{y};E)$,
$\B{v}\in\boxx_L(\B{u})$, with a fixed $\B{y}\in\boxx_L(\B{u})$.

%---------------------------------------------------------------------%
%:sec.4
%---------------------------------------------------------------------%
\section{The $N$-particle MSA induction scheme}     \label{sec:MSA}

In view of Theorem \ref{thm:two}, our aim is to check property $\DSmkn$,
i.e., \eqref{eq:DSmkn}, for $n=N$. As was mentioned before,
it is done by means of a combined induction, in both $k$ and $N$.
Consequently, in some definitions
below we refer to the particle number parameter $n$, whereas in
other definitions - where we want to stress the passage from $N-1$ to
$N$ - we will use the capital letter.

The reader may assume from the start that the interval $I$ is of
ther form $[E^0,E^0+\eta ]$.

We begin with the so-called Wegner-type bounds.
%---------------------------------------------------------------------%
%:s.4.1
%---------------------------------------------------------------------%
\subsection{Wegner-type bounds}  \label{ssec:wegner}

Given $n=1,\ldots ,N$, $q>0$ and $L_0\geq 2$, define two properties
$\Wone{n}$($=${\bf{W1}}$(n,q,L_0)$) and $\Wtwo{n}$($=${\bf{W2}}$(n,q,L_0)$),
for random $n$-particle Hamiltonians
$\B{H}_{\boxl}^{(n)}$ where $\boxl=\boxl_l^{(n)}(\B{x})$
and $l\geq L_0$.

%-----------------------%
\begin{enumerate}[$\Wone{n}$:]
\item
For all $l\geq L_0$, for all $\B{x}\in\D{R}^{nd}$ and for
all $E\in\D{R}$,
\begin{equation}             \label{eq:w-one}
\prob\accol{\, \boxl^{(n)}_l(\B{x})\text{ is not }E\text{-CNR}}<l^{-q}.
\end{equation}
\end{enumerate}
%-----------------------%

%-----------------------%
\begin{enumerate}[$\Wtwo{n}$:]
\item
Given a bounded interval $I\subset\D{R}$, for all $l\geq L_0$ and
any separable boxes $\boxl^{(n)}_\ell(\B{x})$ and
$\boxl^{(n)}_\ell(\B{y})$,
\begin{equation}  \label{eq:w-two}
\prob\accol{\, \text{for some }E\in\D{R},\text{ neither }
\boxl^{(n)}_l(\B{x})\text{ nor }\boxl^{(n)}_l(\B{y})\text{ is }
E\text{-CNR}}<l^{-q}.
\end{equation}
\end{enumerate}
%-----------------------%

%-----------------------%
\begin{theorem} \label{thm:Wegner}
For any $q>0$ and a bounded interval $I\subset\D{R}$,
there exists $L^*_{\rm W} = L^*_{\rm W}(q,I)\in (0,+\infty )$ such that
$\Wone{n}$ and $\Wtwo{n}$ hold true for all $n=1,\ldots ,N$ and $L_0\ge L^*_{\rm W}$.
\end{theorem}
%-----------------------%
\proof See \cite{BCSS010}.
%---------------------------------------------------------------------%

\subsection{The initial step}  \label{ssec:initial.step}
The initial step of the MSA induction
consists in establishing properties $\DSmzeron$ below,
for $n=1,\ldots ,N$:

%-----------------------%
\begin{enumerate}[$\DSmzeron$:]
\item
$\forall$ pair of separable boxes $\boxl_{L_0}^{(n)}(\B{u})$,
$\boxl_{L_0}^{(n)}(\B{v})$,
\begin{equation}  \label{eq:DSmzeronl}        %\eqno(1.20)
\prob\bigl\{
\;\hbox{both }\;\boxl_{L_0}^{(n)}(\B{u})
\text{ and }\boxl_{L_0}^{(n)}(\B{v})\text{ are }
(E,m)\hbox{-S for some }E\in I\;\bigr\}<L_0^{-2p}.
\end{equation}
\end{enumerate}
%-----------------------%

We summarise it in Theorem \ref{thm:MSAInd0}:

%-----------------------%
\begin{theorem}\label{thm:MSAInd0}
Let $m>0$ and a positive integer $L_0$ be given. Then $\;\exists\;$ a
value $\eta^*_0=\eta^*_0(m,L_0)>0$ with the following property.
\begin{enumerate}[\rm(i)]
\item
There exists a function
$$(Nd,+\infty)\ni p\mapsto \eta_0(p)\in(0,\eta^*_0),\;\hbox{ with }
\eta_0(p)\searrow 0\;\hbox{ as }\;p\nearrow +\infty,$$
such that $\forall$ $p>Nd$, Eqn \eqref{eq:DSmzeronl} is satisfied
$\forall$ $n=1,\ldots ,N$ with $I= [E^0, E^0+\eta_0(p)]$.
\item
Equivalently, there exists a function
$$
(0,\eta^*_0)\ni\eta\mapsto p_0(\eta)>Nd,\;\hbox{ with }
p_0(\eta)\nearrow\infty\;\hbox{ as }\;\eta\searrow 0,
$$
such that $\forall$ $\eta\in (0,\eta^*_0)$, Eqn
\eqref{eq:DSmzeronl} is fulfilled $\forall$ $n=1,\ldots ,N$,
with $p=p_0(\eta)$ and $I= [E^0, E^0+\eta ]$.
\end{enumerate}
\end{theorem}
%-----------------------%

\begin{proof}
The assertion of Theorem \ref{thm:MSAInd0} follows directly from
\cite{St01}*{Theorems 2.2.3, 3.3.3} and is omitted from the
paper. It is instructive to observe that the proofs in \cite{St01}
do not rely upon a single- or multi-particle structure
of the potential.
\end{proof}
%-----------------------%

%---------------------------------------------------------------------%
%:s.4.3
%---------------------------------------------------------------------%
\subsection{The inductive step}  \label{ssec:initial.step}
The inductive step of the MSA induction consists in deducing,
given $k\geq 0$, property
$\DSmkponeN$ from from properties
$\DSmkn$ assumed for all $n=1,\ldots ,N$ and  properties
$\DSmkponen$ assumed for all $n=1,\ldots ,N-1$.
Let us summarise:

%-----------------------%
\begin{theorem}\label{thm:MSAInd}
There exist values $L^*_+>0$, $\eta^*_+>0$, and two functions
$$\eta\in (0,\eta^*_+)\mapsto p_+(\eta )>dN,\;\hbox{ with }\;
p_+(\eta)\xrightarrow[\eta\searrow 0]{}+\infty , $$
$$\eta\in (0,\eta^*_+)\mapsto m_+(\eta )>0,\;\hbox{ with }\;
m_+(\eta)\xrightarrow[\eta\searrow 0]{}0, $$
with the following
property. $\forall$ given
$k\geq 0$, suppose that $0<\eta<\eta^*_+$, $L_0\geq L^*_+$ and
\begin{itemize}
\item
property
$\DSmkn$ holds with $I=[E^0,E^0+\eta]$, $m=m_+(\eta )$
and $p=p_+(\eta )$,
$\forall$ $\;n=1,\ldots , N$ and
\item
property
$\DSmkponen$ holds with $I=[E^0,E^0+\eta]$, $m=m_+(\eta )$
and $p=p_+(\eta )$,
$\forall$ $\;n=1,\ldots , N-1$.
\end{itemize}
Then $\;\DSmkponeN\;$ also holds with $I=[E^0,E^0+\eta]$,
$m=m_+(\eta )$ and $p=p_+(\eta )$.
\end{theorem}
%-----------------------%

%-----------------------%
The rest of the paper is devoted to the proof of Theorem
\ref{thm:MSAInd}. Observe that once this proof is completed,
Theorem \ref{thm:three} and hence Theorem \ref{thm:main}
will be established.

%-----------------------%

%---------------------------------------------------------------------%
%:s.4.4
%---------------------------------------------------------------------%
\subsection{Interactive boxes}  \label{ssec:interactive.box}
Recall: $\R{r}_0\in (0,+\infty )$ is the interaction radius
(cf. \eqref{eq:interaction.finite.range}). Consider the
following subset in $\D{R}^{nd}$, $n=1,\ldots ,N$:
\begin{equation}        \label{eq:}   %\eqno(4.6)
\D{D}^{(n)}=\bigl\{\B{x}=(x_1,\dots,x_n)\in\D{Z}^{nd}:
\;\max_{1\leq j_1,\,j_2 \leq n}  \abs{x_{j_1}-x_{j_2}}\leq N\R{r}_0\bigr\}
\end{equation}
It is plain that, $\forall$ $\B{x}\in\D{Z}^{nd}\setminus
\D{D}^{(n)}$,
\begin{equation}\label{eq:noninteraction}
\exists\;\hbox{non-empty}\;\C{J}\subset\accol{1,\dots,n}\;\hbox{such that}\;
\min_{\substack{j_1\in\C{J}\\
       j_2\not\in\C{J}}}\abs{x_{j_1}-x_{j_2}}>\R{r}_0.
\end{equation}

%-----------------------%
\begin{definition}[Interactive boxes]       \label{def:FIPI}
Let $\boxl=\boxl^{(n)}_L(\B{u})$ be an $n$-particle box. We say that
\begin{enumerate}[(i)]
\item
$\boxl$ is \emph{fully interactive} (FI) when $\boxl\cap\,\D{D}^{(n)}
\neq\varnothing$,
\item
$\boxl$ is \emph{partially interactive} (PI) when
$\boxl\cap\,\D{D}^{(n)}=\varnothing$.
\end{enumerate}
\end{definition}
%-----------------------%

The procedure of deducing property $\DSmkponeN$ from $\DSmkN$
and $\DSmkponen$ with $n=1,\ldots , N-1$ is done
separately for the following three types of pairs
$\boxl_{L_{k+1}}^{(N)}(\B{u})$, $\boxl_{L_{k+1}}^{(N)}(\B{v})$ of
separable boxes:
\begin{enumerate}[(I)]
\item
Both $\boxl_{L_{k+1}}^{(N)}(\B{u})$, $\boxl_{L_{k+1}}^{(N)}(\B{v})$
are PI.
\item
Both $\boxl_{L_{k+1}}^{(N)}(\B{u})$, $\boxl_{L_{k+1}}^{(N)}(\B{v})$
are FI.
\item
One of $\boxl_{L_{k+1}}^{(N)}(\B{u})$, $\boxl_{L_{k+1}}^{(N)}(\B{v})$
is FI, while the other is PI.
\end{enumerate}
These three cases are treated separately in Sections~\ref{sec:caseI},
\ref{sec:caseII} and \ref{sec:caseIII}, respectively. The end
of proof of Theorem \ref{thm:MSAInd} is achieved at the end of
Section \ref{sec:caseIII}.

%---------------------------------------------------------------------%
%:sec.5
%---------------------------------------------------------------------%
\section{Separable pairs of partially interactive singular
boxes}\label{sec:caseI}

In this section, we aim to derive property $\DSmkponeN$ in case (I),
i.e., for a PI pair of separable boxes
$\boxl_{L_{k+1}}^{(N)}(\B{u})$, $\boxl_{L_{k+1}}^{(N)}(\B{v})$.
In this particular case we will be able to do this without
referring to $\DSmkn$. However, we will use properties
$\DSmkponen$ for $n=1,\ldots ,N-1$. Cf. the statement
of Theorem \ref{thm:kplusoneforPI} below.

Let $\boxl=\boxl_{L_{k+1}}^{(N)}(\B{u})$ be a PI-box where
$\B{u}=(u_1,\dots,u_N)\in\D{Z}^{Nd}$. Let $\C{J}\subset
\accol{1,\dots,N}$ be a proper subset figuring in Eqn
\eqref{eq:noninteraction}.
Write $\B{u}=(\B{u}',\B{u}'')$ where $\B{u}'=\B{u}_{\C{J}}
=(u_j)_{j\in\C{J}}\in(\D{Z}^d)^{\C{J}}$ and
$\B{u}''=\B{u}_{\C{J}^{\compl}}=(u_j)_{j\notin\C{J}}
\in(\D{Z}^d)^{\C{J}^{\compl}}$ are the corresponding
sub-configurations in $\B{u}$. Let $n'=\#\C{J}$ be the cardinality
of $\C{J}$ and $n''=N-n'$. We write $\boxl$ as the
Cartesian product
\[\boxl=\boxl'\times\boxl'',\;\hbox{ where }
\boxl'=\boxl_{L_{k+1}}^{(n')}(\B{u}'),\;
\boxl''\boxl_{L_{k+1}}^{(n'')}(\B{u}'').
\]
The Hamiltonian
$\B{H}^{(N)}_{\boxl_{L_{k+1}}^{(N)}(\B{u})}$ can be represented as
\begin{equation}  \label{eq:splitting}
\B{H}_{\boxl}^{(N)}=\B{H}_{\boxl'}^{(n')}
\otimes\B{I}^{(n'')}+\B{I}^{(n')}\otimes\B{H}_{\boxl''}^{(n'')}.
\end{equation}
Here $\B{I}^{(n')}$ and $\B{I}^{(n'')}$ are the identity operators
on $L_2(\boxl')$ and $L_2(\boxl'')$, respectively. A similar
decomposition can be written for each $\sigma(\B{u})=(u_{\sigma(1)},
\dots,u_{\sigma(N)})$, for any permutation $\sigma$ of order
$N$.

%-----------------------%
\begin{definition}[$(I,m)$-partial tunneling]               \label{def:Bad}
In this definition we deal with $m>0$, $1\leq n'\leq N-1$,
$k\geq 0$ and $\B{u}'=(u_1,\dots,u_{n'})
\in\D{Z}^{nd'}$ and $\B{u}=(u_1,\dots,u_N)
\in\D{Z}^{Nd}$ and a bounded interval $I\subset\D{R}$
(eventually, of the form $I=[E^0,E^0+\eta]$).
\begin{enumerate}[\rm (i)]
\item
An $n'$-particle box $\boxl^{(n')}_{L_{k+1}}(\B{u}')$ is
\emph{$(I,m)$-tunneling} ($m$-T) if there exists $E\in I$ and
two separable $n$-particle boxes $\boxl_{L_k}^{(n')}(\B{v}_j)
\subset\boxl_{L_{k+1}}^{(n')}(\B{u}')$, $j=1,2$ which are $(E,m)$-S.
\item
An $N$-particle box $\boxl_{L_{k+1}}^{(N)}(\B{u})$ is
\emph{$(I,m,n',n'')$-partially tunelling} if
\begin{enumerate}[\textbullet]
\item
$n'+n''=N$ and $n',n''\geq 1$,
\item
for $\B{u}= (\B{u}',\B{u}'')$, $\B{u}'
=(u_1,\dots,u_{n'})$, $\B{u}''=(u_{n'+1},\dots,u_{N})$, we have
\[
\boxl_{L_{k+1}}^{(N)}(\B{u})=\boxl_{L_{k+1}}^{(n')}(\B{u}')
\times\boxl_{L_{k+1}}^{(n'')}(\B{u}''),
\]
\item
either $\boxl_{L_{k+1}}^{(n')}(\B{u}')$ or
$\boxl_{L_{k+1}}^{(n')}(\B{u}'')$ is $(I,m)$-tunneling.
\end{enumerate}
\item
An $N$-particle box $\boxl_{L_{k+1}}^{(N)}(\B{u})$ is
\emph{$(I,m)$-partially tunelling} ($(I,m)$-PT) if, for some permutation
$\sigma$, $\boxl_{L_{k+1}}^{(N)}(\sigma(\B{u}))$ is $(I,m,n',n'')$-partially
tunelling for some $n',n''\geq 1$ with $n'+n''=N$. Otherwise, it
is called \emph{$(I,m)$-non partially tunelling} ($(I,m)$-NPT).
\end{enumerate}
\end{definition}
%-----------------------%

%\begin{lemma}          \label{lem:PITRoNS}
\begin{lemma}          \label{lem:PITRoNS}
Consider an $n$-particle box of the form
\[
\boxl_{L_{k+1}}^{(n)}(\B{u})=\boxl_{L_{k+1}}^{(n')}(\B{u}')
\times\boxl_{L_{k+1}}^{(n'')}(\B{u}'')=\boxl'\times\boxl'',
\]
where $\B{u}= (\B{u}',\B{u}'')$, $\B{u}'=(u_1,\dots,u_{n'})
\in\D{Z}^{nd'}$, $\B{u}''=(u_{n'+1},\dots,u_{n})\in
\D{Z}^{nd''}$. Set:
\[
\boxl=\boxl_{L_{k+1}}^{(n)}(\B{u}),\quad \boxl'
=\boxl_{L_{k+1}}^{(n')}(\B{u}'),\quad \boxl''=\boxl_{L_{k+1}}^{(n'')}(\B{u}'').
\]
\begin{enumerate}[\rm(a)]
\item
Assume that $\forall$
$1\leq j_1\leq n'$, $n'+1\leq j_2\leq n$, we have
\[
\abs{y_{j_1}-y_{j_2}}>\R{r}_0,\;\forall\;\B{y}=(y_1,\ldots ,y_N)\in\boxl,
\]
so that $\boxl$ is {\rm{PI}}.
\item
Assume also that box $\boxl$ is $(I,m)$-\emph{NPT} for some $m>0$
and $E$-\emph{CNR} for some $E\in I$ where $I\subset\D{R}$ is
a bounded interval.
\end{enumerate}
Let $(E'_a,\BS{\varPsi}'_a)$ for $a\ge 1$, and $(E''_b,\BS{\varPsi}''_b)$ for
$b\ge 1$ be the eigenvalues and eigenvectors of
$\B{H}^{(n')}_{\boxl'}$ and $\B{H}^{(n'')}_{\boxl''}$, respectively.
Then, for $L_0$ large enough,
the discretized Green functions obey:
\begin{equation}\label{eq:lemma 5.2}
\begin{array}{c}
\max\;\Big[\;\R{D}_{L_{k+1},\B{u}''}(\B{u}'',\B{v}'';E-E'_a)\,:\;
a \ge 1,\;\B{v}''\in\partial^-\boxl''\;\Big]
\leq\ee^{-\gamma(m,L_{k+1},n)}\,,\\
\max\;\Big[\;\R{D}_{L_{k+1},\B{u}'}(\B{u}',\B{v}';E-E''_b)\,:\;
b \ge 1,\;\B{v}'\in\partial^-\boxl'\;\Big]
\leq\ee^{-\gamma(m,L_{k+1},n)}\,.
\end{array}
\end{equation}
This implies that box $\boxl$ is $(E,m)$-\emph{NS}.
\end{lemma}
%-----------------------%

%-----------------------%
\begin{proof}
The proof is given in Section \ref{sec:appendix}.
\end{proof}
%-----------------------%

%-----------------------%
\begin{lemma}\label{lem:one}
Given $m>0$, $p>0$, a bounded interval $I\subset\D{R}$ and
$n=1,\ldots ,N-1$, suppose that
property $\DSmkponen$ holds for
some $k\geq 0$. Then, for any $\B{u}\in\D{Z}^{nd}$,
\begin{equation}  \label{eq:lemone}        %\eqno(3.3)
\prob\bigl\{\boxl_{L_{k+1}}^{(n)}(\B{u})\text{ is }m
\textup{-PT}\bigr\}\leq\frac{1}{2}\abs{\boxl^{(n)}_{L_{k+1}}(\B{u})}^2
\times L_k^{-2p}=\frac{1}{2}L_{k+1}^{-2p/\alpha +2d}.
\end{equation}
\end{lemma}
%-----------------------%

%-----------------------%
\begin{proof}
Combine $\DSmkponen$ with a straightforward (albeit not sharp) upper
bound $\frac{1}{2}\abs{\boxl^{(n)}_{L_{k+1}}(\B{u})}^2$ for the
number of pairs $(\B{y}_1,\,\B{y}_2)$ of centers of boxes
$\boxl^{(n)}_{L_k}(\B{y}_j)\subset\boxl^{(n)}_{L_{k+1}}(\B{u})$, $j=1,2$.
\end{proof}
%-----------------------%

%-----------------------%
\begin{lemma}\label{lem:two}
Suppose $m>0$ and a bounded interval $I\subset\D{R}$
have been given.
Let $\boxl=\boxl_{L_{k+1}}^{(N)}(\B{u})$,
$\B{u}=(u_1,\dots,u_N)\in\D{Z}^{Nd}$, be an $N$-particle
\emph{PI}-box of the form
$$
\boxl =\boxl'\times\boxl''\;\hbox{ with }
\boxl'=\boxl^{(n')}_{L_{k+1}}(\B{u}'),\;\boxl''=
\boxl_{L_{k+1}}^{(n'')}(\B{u}''),$$
where $n'+n''=N$, $n',n''\geq 1$, $\B{u}=(\B{u}',\B{u}'')$,
$\B{u}'= (u_1,\dots, u_{n'})$, $\B{u}''=(u_{n'+1},\dots, u_N)$. Assume
that $\forall$ $\B{y}=(y_1,\ldots ,y_N)\in\boxl$,
\[
\min_{\substack{1\leq i\leq n'\\
                n'+1\leq j\leq N}}\abs{y_i-y_j}>\R{r}_0.
\]
Then for any $p>0$ there exists $\eta^*_{\rm{PT}}\in(0,+\infty)$ such that
the condition $0<\eta \leq \eta^*_{\rm{PT}}$ implies that
\begin{equation}  \label{eq:lemtwo}        %\eqno(3.4)
\prob\bigl\{\boxl\text{ is
}m\textup{-PT}\bigr\}\leq\frac{1}{4}L_{k+1}^{-2p}.
\end{equation}
\end{lemma}
%-----------------------%

%-----------------------%
\begin{proof}
By Definition \ref{def:Bad}, $\boxl$ is $m$--PT if and only if at
least one of the boxes $\boxl'$ or $\boxl''$ is $m$-T. By
Lemma \ref{lem:one}, Eqn \eqref{eq:lemone} holds for both $n=n'$ and $n=n''$.
Since parameter $p_0(\eta)\to\infty$ as $\eta\to 0$ (see Theorem
\ref{thm:MSAInd0}), this leads to
the assertion of Lemma \ref{lem:two}.
\end{proof}
%-----------------------%

%-----------------------%
\begin{lemma}             \label{lem:TwoOffDiagSin}
Given $L_0\geq 1$, $m>0$, $q>0$, $p > \alpha(p_0(\eta )+d)$ and a bounded interval
$I\subset\D{R}$, assume that
\begin{itemize}
\item
the bound \eqref{eq:lemtwo} holds true,
\item
for all $n=1, \dots, N-1$
the bound \eqref{eq:lemone} holds,
\item
$L_0$ is sufficiently large, so that for any $k\geq 0$ we have
\[
L_k^{-2p/\alpha +2d}\leq\frac{1}{4}L_k^{-2p_0(\eta )},\]
\item
the bound \eqref{eq:w-two} with $n=N$ (i.e., property $\Wtwo{N}$)
is satisfied.
\end{itemize}
Then, for any integer $k\geq 0$ and
for any pair of separable, \emph{PI}
$N$-particle boxes $\boxl_{L_k}^{(N)}(\B{u})$ and $\boxl_{L_k}^{(N)}(\B{v})$,
\begin{equation}  \label{eq:TwoOffDiagSing}        %\eqno(3.5)
\prob\accol{\boxl^{(N)}_{L_{k+1}}(\B{u}),
\boxl^{(N)}_{L_{k+1}}(\B{v})\text{ are }(E,m)\text{\emph{-S}
for some }E\in I}\leq\frac{1}{2}L_{k+1}^{-2p_0(\eta )}+L_{k+1}^{-q}.
\end{equation}
\end{lemma}
%-----------------------%

%-----------------------%
\begin{proof}
Set $\boxl(\B{u})=\boxl_{L_{k+1}}^{(N)}(\B{u})$ and
$\boxl(\B{v})=\boxl_{L_{k+1}}^{(N)}(\B{v})$. Lemma \ref{lem:PITRoNS} implies
\begin{equation}  \label{eq:OffDiagSing}        %\eqno(3.6)
\begin{array}{l}
\prob\accol{\boxl(\B{u})\text{ and }\boxl(\B{v})
\text{ are }(E,m)\text{-S for some }E\in I}\\
\qquad\leq\prob\accol{\boxl(\B{u})\text{ or }
\boxl(\B{v})\text{ is }\text{-PT}}\\
\qquad\quad +\prob\accol{\text{neither }\boxl(\B{u})
\text{ nor }\boxl(\B{v})\text{ is }E\text{-CNR}\text{ for some }E\in I}.
\end{array}\end{equation}
The assertion now follows
from the assumptions  of Lemma \ref{lem:TwoOffDiagSin} and from the statement of Lemma \ref{lem:two}.
\end{proof}
%-----------------------%

\begin{theorem} \label{thm:kplusoneforPI}

Given $p^*>Nd$, there exist $m^*_{\rm{PI}}>0$,
$\eta^*_{\rm{PI}}>0$ and a positive $L^*_{\rm{PI}}<+\infty$ with the
following property. Take $L_0\geq L^*_{\rm{PI}}$.
Then, $\forall$  $k\geq 0$, $\DSmkponeN$ holds for all
separable pairs of $N$-particle
\emph{PI}-boxes
$\boxl_{L_{k+1}}^{(N)}(\B{u})$, $\boxl_{L_{k+1}}^{(N)}(\B{v})$ with
$m=m^*_{\rm{PI}}$, $p=p^*$ and
interval $I=[E^0,E^0+\eta^*_{\rm{PI}}]$.
\end{theorem}
%-----------------------%

%-----------------------%
\begin{proof}
The statement of Theorem \ref{thm:kplusoneforPI} is
an immediate corollary of Theorem \ref{thm:Wegner} and
Lemma \ref{lem:TwoOffDiagSin}.
\end{proof}
%-----------------------%

For future use, we also give

%-----------------------%
\begin{lemma}                \label{lem:onnuPI}
Given $0<\eta<\min\;\big[\eta^*_0,\eta^*_{\rm{PI}}\big]$,
$L_0\geq 1$, $q>0$, $p\geq 2p_0(\eta )+2d$ and a bounded interval
$I\subset\D{R}$, assume that
\begin{itemize}
\item
the bound \eqref{eq:lemtwo} holds true,
\item
for all $n=1, \dots, N-1$
the bound \eqref{eq:lemone} holds,
\item
$L_0$ is sufficiently large, so that for any $k\geq 0$ we have
\[
L_k^{-2p/\alpha +2d}\leq\frac{1}{4}L_k^{-2p_0(\eta )},\]
\item
the bound \eqref{eq:w-one} with $n=N$ (i.e., property $\Wone{N}$)
is satisfied.
\end{itemize}
Let $\boxl=\boxl_{L_{k+1}}^{(N)}(\B{u})$ be an $N$-particle box. Let
$\nu_{\rm{PI}}(\boxl;E)$ be the (random) maximal number of $(E,m)$-\emph{S},
pairwise separable
\emph{PI-}boxes $\boxl_{L_k}^{(N)}(\B{y})\subset\boxl$.
Then the following inequality takes place:
\begin{equation}  \label{eq:onnuPI}        %\eqno(3.8)
\prob\bigl\{\nu_{\rm{PI}}(\boxl;E)\geq 2\text{ for some }E\in I\bigr\}
\leq \half L_k^{2d\alpha}
\Bigl(L_k^{-2p_0(\eta )}+L_{k}^{-q}\Bigr).
\end{equation}
\end{lemma}
%-----------------------%

%-----------------------%
\begin{proof}
If $\nu_{\rm{PI}} \ge 2$, then there exist (at least) two singular boxes
$\boxl_{L_k}(\B{x}),\boxl_{L_k}(\B{y})$.  The number of possible pairs $(\B{x},\B{y})$
is bounded by $\half L_{k+1}^{2d}$, while
for a given pair $\boxl_{L_k}(\B{x})$, $\boxl_{L_k}(\B{y})$
Lemma
\ref{lem:TwoOffDiagSin} applies. This leads to the assertion
of Lemma \ref{lem:onnuPI}.
\end{proof}
%-----------------------%

%---------------------------------------------------------------------%
%:sec.6
%---------------------------------------------------------------------%
\section{Separable pairs of fully interactive
singular boxes}\label{sec:caseII}

The main outcome in case (II) is Theorem \ref{thm:kplusoneforFI} at the
end of this section. Recall, the definition of an FI-box
was related to
$\R{r}_0\in (0,+\infty )$, the radius of interaction
(cf. \eqref{eq:interaction.finite.range}). Further, the definition
of a separable pair of boxes $\boxl_L(\B{u})$ and
$\boxl_L(\B{v})$ was related to the constant $R$, the diameter
of support of the bump functions
and included the condition
$$\dist\;\left(\boxl_L(\B{u}),\boxl_L(\B{v})\right)>2N(L+R)$$
(see Definition \ref{def:separable.boxes}).
Before we proceed further, let us state a geometric assertion:

%-----------------------%
\begin{lemma}\label{lem:DistDiag}
Let $L>\R{r}_0$ be an integer. Let
$\boxl_L(\B{u}')$
and $\boxl_L(\B{u}'')$ be two separable $N$-particle
\emph{FI}-boxes, where $\B{u}' = (u_1',\dots,u_N')$,
$\B{u}'' = (u_1'',\dots,u_N'')$. Then
\begin{equation}  \label{eq:DistDiag}        %\eqno(4.1)
\varPi\boxl_{L+R}(\B{u}')\cap\varPi\boxl_{L+R}(\B{u}'')=\varnothing.
\end{equation}
\end{lemma}
%-----------------------%

%-----------------------%
\begin{proof}
If $\boxl_L(\B{u}')$ is FI, then there exists a permutation
$\sigma$ of order $N$ such that, for all $j=1,\dots,N-1$,
\[
\abs{u_j'-u_{j+1}'}\leq \R{r}_0.
\]
Otherwise, the set $\accol{u_j'}_{1\leq j\leq N}\subset\D{Z}^d$ could be
decomposed into two or more non-interacting subsets. Therefore,
$$\diam \accol{u_j'}_{1\leq j\leq n}\leq (N-1)\R{r}_0;\hbox{ similarly, }
\diam \accol{u_j''}_{1\leq j\leq n}\leq (N-1)\R{r}_0.$$
Further,
suppose that for some $i,j\in\accol{1,\dots,n}$, we have
\[
\varPi_i\boxl_{L+R}(\B{u}')\cap\varPi_j\boxl_{L+R}(\B{u}'')\neq\varnothing.
\]
Then $\abs{u_i'-u_j''}\leq 2(L+R)$, and, therefore, for any $k=1,\dots,n$
\begin{align*}
\abs{u_k'-u_k''}
&\leq\abs{u_k'-u_i'}+\abs{u_i'-u_j''}+\abs{u_j''-u_k''}\\
&\leq (N-1)\R{r}_0+2(L+R)+(N-1)\R{r}_0\\
&\leq 2N(L+R).
\end{align*}
This is incompatible with the inequality
$\dist(\boxl_L(\B{u}'),\boxl_L(\B{u}'')) > 2nN(L+R)$, since in the
latter case there must exist some $k$ such that
$\abs{u_k'-u_k''}> 2(L+R)$.
\end{proof}
%-----------------------%

Lemma \ref{lem:DistDiag} is used in the proof of Lemma \ref{lem:onnuFI}
which, in turn, is a part of the proof of Lemma \ref{lem:onnuS},
instrumental in establishing Theorem \ref{thm:kplusoneforFI}.

Let an interval $I\subset\D{R}$ and a number $m>0$ be given. Consider the
following assertion
which is a particular case of $\DSmkN$ (cf. Eqn \ref{eq:DSmkn}):
\begin{enumerate}[$\FISk$:]
\item
For any pair of separable $N$-particle FI-boxes
$\boxl_{L_k}^{(N)}(\B{u})$ and
$\boxl_{L_k}^{(N)}(\B{v})$
\begin{equation}  \label{eq:FISk}        %\eqno(4.2)
\prob\bigl\{\boxl_{L_k}^{(N)}(\B{u})\text{ and }\boxl_{L_k}^{(N)}(\B{v})
\text{ are }(E,m)\text{-S for some }E\in I\bigr\}\leq L_k^{-2p}.
\end{equation}
\end{enumerate}

%-----------------------%
\begin{lemma}                 \label{lem:onnuFI}
Let $k\geq 0$ be given. Assume that property $\FISk$ holds true. Let
$\boxl=\boxl_{L_{k+1}}^{(N)}(\B{u})$ be an $N$-particle box. Denote
by $\nu_{\rm{FI}}(\boxl;E)$ the (random) maximal number of $(E,m)$-{\rm S}, pairwise
separable \emph{FI}-boxes $\boxl_{L_k}^{(N)}(\B{y}^{(j)})\subset\boxl$.
Then, for any $\ell\geq 1$,
\begin{equation}  \label{eq:ProbISing}        %\eqno(4.3)
\prob\Big\{\nu_{\rm{FI}}(\boxl;E)\geq 2\ell\text{ for some }E\in I\bigr\}
\leq L_k^{2\ell(1+d\alpha)??}\cdot L_k^{-2\ell p}.
\end{equation}
\end{lemma}
%-----------------------%

\begin{proof}
Suppose there exist FI-boxes $\boxl_{L_k}^{(N)}(\B{y}^{(j)})
\subset\boxl$, $j=1,\dots,2\ell$, such that any two of them are
separable. By virtue of Lemma \ref{lem:DistDiag}, it is readily
seen that
the pairs of operators $\left(\B{H}_{\boxl_{L_k}(\B{y}^{(2i-1)})},
\B{H}_{\boxl_{L_k}(\B{y}^{(2i)})}\right)$, $i=1,\ldots ,\ell$,
form an independent family. [It is also true that, within a given pair,
operators $\B{H}_{\boxl_{L_k}(\B{y}^{(2i-1)})}(\omega)$ and
$\B{H}_{\boxl_{L_k}(\B{y}^{(2i)})}(\omega)$ are mutually
independent.]

Thus, any collection of events $\C{A}_1,\dots,\C{A}_{\ell}$ related to
these pairs also forms an independent family. Now, for $i=1,\dots,\ell$,
set
\begin{equation}  \label{eq:}        %\eqno(4.5)
\C{A}_i=\Big\{\boxl_{L_k}(\B{y}^{(2i-1)})
\text{ and }\boxl_{L_k}(\B{y}^{(2i)})\text{ are }(E,m)\text{-S
for some }E\in I\Big\}.
\end{equation}
Then, owing to property $\FISk$ (see \eqref{eq:FISk}), $\forall$
$i=1,\dots,\ell$,
\begin{equation}  \label{eq:}        %\eqno(4.6)
\prob\left\{ \C{A}_i \right\} \leq L_k^{-2p},
\end{equation}
and by virtue of independence of events $\C{A}_1,\dots,\C{A}_{\ell}$,
we obtain that
\begin{equation}  \label{eq:}        %\eqno(4.7)
\prob\left\{ \bigcap_{i=1}^{\ell}\C{A}_i \right\} =
\prod_{i=1}^{\ell}\prob\left\{\C{A}_i \right\}
\leq\bigl(L_k^{-2p}\bigr)^{\ell}.
\end{equation}
To complete the proof, note that the total number of different families
of $2\ell$ boxes $\boxl_{L_k}^{(N)}\subset\boxl_{L_{k+1}}^{(N)}(\B{u})$
with required properties is bounded from above by
$$
\frac{1}{(2\ell)!} \left[ 2 (L_k + \R{r}_0 + 1) L_{k+1}^d\right]^{2\ell}
\leq \frac{1}{(2\ell)!} \left(3 L_k L_{k+1}^d\right)^{2\ell}
\leq L_k^{2\ell(1+d\alpha)}.
$$
In fact, their centres must lie at distance $\leq L_k+\R{r}_0$ from the
set $\D{D}^{(N)}_{}\cap \boxx_{L_{k+1}}^{(N)}(\B{u})$.
This yields the assertion of Lemma \ref{lem:onnuFI}.

\end{proof}
%-----------------------%

\begin{lemma}          \label{lem:onnuS}
Given $k\geq 0$, let $\boxl=\boxl_{L_{k+1}}^{(N)}(\B{u})$ be
an $N$-particle box.
Consider an interval $I$ of the form $I=[E^0,E^0+\eta ]$ and assume
that the conditions of Lemmas \ref{lem:onnuPI} and \ref{lem:onnuFI}
are fufilled.
Given $E\in I$, denote by $\nu_{\rm S}(\boxl ;E)$ the (random) maximal
number of $(E,m)$-\emph{S}, pairwise separable
boxes $\boxl_{L_k}^{(N)}(\B{u}^{(j)})\subset\boxl$. Let
$\kappa (N)$ be the constant from Lemma \ref{lem:CondGeomSep}.
Then, $\forall$ $\ell\geq 1$,
\begin{align}  \label{eq:lemonnuS}        %\eqno(4.8)
&\prob\accol{\nu_{\rm S}(\boxl ;E)\geq 2\ell+\kappa (N)+1
\text{ for some }E\in I}\notag\\
&\qquad\qquad\leq L_k^{4d\alpha}\cdot L_k^{-2 p_0(\eta )}+
L_k^{2\ell(1+d\alpha)??}\cdot L_k^{-2\ell p}.
\end{align}
\end{lemma}
%-----------------------%

\begin{proof}
Suppose that $\nu_{\rm S}(\boxl ;E)\geq 2\ell+\kappa (N)+1$. Let
$\nu_{\rm{PI}}(\boxl;E)$
be as in Lemma \ref{lem:onnuPI} and $\nu_{\rm{FI}}(\boxl;E)$ as
in Lemma \ref{lem:onnuFI}.
Obviously,
$$
\nu_{\rm S}(\boxl ;E)\leq \nu_{\rm{PI}}(\boxl;E)+\nu_{\rm{FI}}(\boxl;E).
$$
Then either $\nu_{\rm{PI}}(\boxl;E)\geq\kappa (N)+1$ or $\nu_{\rm{FI}}
(\boxl;E)\geq 2\ell$. Therefore,
\begin{align*}
&\prob\accol{\nu_{\rm S}(\boxl ;E)\geq 2\ell+ \kappa(N) + 1 \text{ for some }E\in I}\\
&\qquad\qquad\leq\prob\accol{\nu_{\rm{PI}}(\boxl;E)\geq  \kappa(N) + 1\text{ for some }
E\in I}\\
&\qquad\qquad\quad
+\prob\accol{\nu_{\rm{FI}}(\boxl;E)\geq 2\ell\text{ for some }E\in I}\\
&\qquad\qquad\leq L_k^{4d\alpha}\cdot L_k^{-2 p_0(\eta )}
+L_k^{2\ell(1+d\alpha)}\cdot L_k^{-2\ell p},
\end{align*}
by virtue of \eqref{eq:onnuPI} and \eqref{eq:ProbISing}.
\end{proof}
%-----------------------%

An elementary calculation gives rise to the following

%-----------------------%
\begin{corollary}\label{cor:onnuS}
Under assumptions of Lemma \emph{\ref{lem:onnuS}}, with
$\ell\geq 2$, $p_0(\eta )$
and $p$ large enough and for $L_0$ large enough, we have,
for any integer $k\geq 0$,
\begin{equation}  \label{eq:coronnuS}        %\eqno(4.9)
\prob\accol{\nu_{\rm S}(\boxl ;E)\geq 2\ell+2\text{ for some }E\in I}
\leq L_{k+1}^{-2p-1}.
\end{equation}
\end{corollary}
%-----------------------%

Now, if two $N$-particle boxes $\boxl_{L_{k+1}}^{(N)}(\B{u}')$
and $\boxl_{L_{k+1}}^{(N)}(\B{u}'')$ are separable, then property
$\Wtwo{N}$ (i.e., Eqn \eqref{eq:w-two} with $n=N$) implies the bound
\begin{align}  \label{eq:coronnuS}        %\eqno(4.10)
&\prob\accol{\text{for any }E\in I,\text{ either }
\boxl_{L_{k+1}}^{(N)}(\B{u}')\text { or }
\boxl_{L_{k+1}}^{(N)}(\B{u}'')\text{ is }E\text{-CNR}}\notag\\
&\qquad\qquad\qquad\qquad\geq 1-L_{k+1}^{-(q\alpha^{-1}-2\alpha)}
>1-L_{k+1}^{-(q'(N)-4)}.
\end{align}
Here $q':=q/\alpha$.

The main result of this section is the following

%-----------------------%
\begin{theorem}     \label{thm:kplusoneforFI}

For any $p^*>Nd$ large enough, there exist $m^*_{\rm{FI}}>0$,
$\eta^*_{\rm{FI}}>0$ and $L^*_{\rm{FI}}\in (0,+\infty)$
such that the following property holds true.
Given $L_0\geq L^*_{\rm{FI}}$ and  $k\geq 0$,
assume that property $\FISk$ holds with $m=m^*_{\rm{FI}}$,
$p=p^*$ and interval $I=[E^0,E^0+\eta^*_{\rm{FI}}]$.
Then property $\FISkone$ also holds, again with
$m=m^*_{\rm{FI}}$, $p=p^*$ and interval $I=[E^0,E^0+\eta^*_{\rm{FI}}]$.
\end{theorem}
%-----------------------%

\begin{proof}
Let $m>0$, $\B{u},\B{v}\in\D{Z}^{Nd}$ and assume that
$\boxl_{L_{k+1}}^{(N)}(\B{u})$ and $\boxl_{L_{k+1}}^{(N)}(\B{v})$ are
separable FI-boxes. With an interval $I$ of the form
$[E^0,E^0+\eta ]$, consider the following two events:
\begin{align*}
\C{B}&=\Bigl\{\boxl_{L_{k+1}}^{(N)}(\B{u})\text{ and }
\boxl_{L_{k+1}}^{(N)}(\B{v})\text{ are }(E,m)\text{-S for some }
E\in I\Bigr\},\\
\C{D}&=\Bigl\{\text{for some }E\in I,\text{ neither }
\boxl_{L_{k+1}}^{(N)}(\B{u})\text{ nor }\boxl_{L_{k+1}}^{(N)}
(\B{v})\text{ is }E\text{-CNR}\Bigr\}.
\end{align*}

The argument that follows assumes that parameters $m$, $\eta$, $p$ and $L_0$
are adjusted in the way specified in the conditions of
Theorem \ref{thm:kplusoneforFI}.
Owing to property $\Wtwo{N}$ (cf. Eqn \eqref{eq:w-two}, with $n=N$),
we have:
\begin{equation}  \label{eq:prob.r}        %\eqno(4.12)
\prob(\C{D})<L_{k+1}^{-(q'-4)},\text{ where }q':=\frac{q}{\alpha}\,.
\end{equation}
Moreover, $\prob(\C{B})\leq\prob(\C{D})+\prob(\C{B}\cap\,\C{D}^{\compl})$.
So, it suffices
to estimate the probability $\prob(\C{B}\cap\,\C{D}^{\compl})$. Within
the event $\C{B}\cap\,\C{D}^{\compl}$, for any $E\in I$, either
$\boxl_{L_{k+1}}^{(N)}(\B{u})$ or $\boxl_{L_{k+1}}^{(N)}(\B{v})$ must
be $E$-CNR. Without loss of generality, assume that for some
$E\in I$, $\boxl_{L_{k+1}}^{(N)}(\B{u})$ is $E$-CNR and $(E,m)$-S.
By Lemma \ref{lem:RadialGF}, if $L_0$ (and, therefore, any $L_k$) is
sufficiently large, for such value of $E$,
$\nu_{\rm S}(\boxl_{L_{k+1}}^{(N)}(\B{u});E)\geq K+1$, with $K$ as in Lemma
\ref{lem:RadialGF}. Now let $K=\kappa (N)$, where $\kappa (N)$
is the contant from Lemma \ref{lem:CondGeomSep}.
We see that
\[
\C{B}\cap\,\C{D}^{\compl}\subset\bigl\{\nu_{\rm S}(\boxl_{L_{k+1}}^{(N)}(\B{u});E)
\geq \kappa (N)+1\text{ for some }E\in I\bigr\}
\]
and, therefore, by Lemma \ref{lem:onnuS} and Corollary
\ref{cor:onnuS},
\begin{equation}  \label{eq:prob.b.compl.r}        %\eqno(4.13)
\prob(\C{B}\cap\,\C{D}^{\compl})\leq\prob
\accol{\exists\, E\in I\,|\,\nu_{\rm S}(\boxl_{L_{k+1}}^{(N)}(\B{u});E)\geq \kappa (N)+1 }
\leq L_k^{-2p}.
\qedhere
\end{equation}
\end{proof}
%-----------------------%

%---------------------------------------------------------------------%
%:sec.7
%---------------------------------------------------------------------%
\section{Mixed separable pairs of singular boxes}\label{sec:caseIII}

It remains to derive property $\DSmkponeN$ in case (III), i.e., for
mixed pairs of
$N$-particle boxes (where one is FI and the other PI).

A natural counterpart of Theorem \ref{thm:kplusoneforFI} for
mixed pairs of boxes is the following

%-----------------------%
\begin{theorem}    \label{thm:kplusoneforMI} % Theorem 7.1

For any $p^*>Nd$ large enough, there exist $m^*_{\rm{MI}}>0$,
$\eta^*_{\rm{MI}}>0$ and $L^*_{\rm{MI}}\in (0,+\infty)$
guaranteing the following property.
Given $L_0\geq L^*_{\rm{FI}}$ and  $k\geq 0$,
assume that property $\DSmkN$ holds, with $m=m^*_{\rm{MI}}$,
$p=p^*$ and interval $I=[E^0,E^0+\eta^*_{\rm{MI}}]$,
\begin{itemize}
\item
for any pair of separable \emph{PI}-boxes $\boxl_{L_k}^{(N)}
({\B{x}})$, $\boxl_{L_k}^{(N)}({\B{y}})$, $\B{x},\B{y}\in\D{Z}^{Nd}$,
\item
for any pair of separable \emph{FI}-boxes
$\boxl_{L_k}^{(N)}(\tilde{\B{x}})$, $\boxl_{L_k}^{(N)}(\tilde{\B{y}})$,
$\B{x},\B{y}\in\D{Z}^{Nd}$.
\end{itemize}
Then property $\DSmkponeN$ holds for mixed pairs of separable boxes
$\boxl_{L_{k+1}}^{(N)}(\B{u})$ and $\boxl_{L_{k+1}}^{(N)}(\B{v})$.

In other words, if $\boxl_{L_{k+1}}^{(N)}(\B{u})$,
$\boxl_{L_{k+1}}^{(N)}(\B{v})$ is a mixed
pair of separable boxes then, for $p=p^*$, $m=m^*_{\rm{MI}}$ and
$I=[E^0,E^0+\eta^*_{\rm{MI}}]$,
\begin{equation}  \label{eq:INIsing}        %\eqno(5.1)
\prob\bigl\{\boxl_{L_{k+1}}^{(N)}(\B{x})
\text{ and }\boxl_{L_{k+1}}^{(N)}(\B{y})
\text{ are }(E,m)\text{-S for some }E\in I\bigr\}\leq L_{k+1}^{-2p}.
\end{equation}
\end{theorem}
%-----------------------%

%-----------------------%
\begin{proof}
Assume that $\boxl_{L_{k+1}}^{(N)}(\B{u})$, $\boxl_{L_{k+1}}^{(N)}(\B{v})$
is separable pair where box $\boxl_{L_{k+1}}^{(N)}(\B{u})$
is FI and $\boxl_{L_{k+1}}^{(N)}(\B{v})$ PI.
Consider the following three events:
\begin{align*}
\C{B}&=\bigl\{\exists\;E\in I:\;\boxl_{L_{k+1}}^{(N)}(\B{u})
\text{ and }\boxl_{L_{k+1}}^{(N)}(\B{v})\text{ are }(E,m)\text{-S}\bigr\},\\
\C{T}&=\bigl\{\boxl_{L_{k+1}}(\B{v})\text{ is }(I,m)\text{-PT}\bigr\},\\
\C{D}&=\bigl\{\exists\;E\in I:\;\text{ neither }\boxl_{L_{k+1}}^{(N)}(\B{u})
\text{ nor }\boxl_{L_{k+1}}^{(N)}(\B{v})\text{ is }E\text{-CNR}\bigr\}.
\end{align*}
By virtue of (3.4??),
\[
\prob(\C{T})\leq\frac{1}{4}L_{k+1}^{-2p},
\]
and by Theorem \ref{thm:Wegner},
\[
\prob(\C{D})\leq L_{k+1}^{-q+2}.
\]
Further,
\[
\prob(\C{B})\leq\prob(\C{T})+\prob(\C{B}\cap\,\C{T}^{\compl})
\leq\frac{1}{4}L_{k+1}^{-2p}+\prob(\C{B}\cap \C{T}^{\compl}).
\]
Thus, we have
\[
\prob(\C{B}\cap\,\C{T}^{\compl})\leq\prob(\C{D})+
\prob(\C{B}\cap\,\C{T}^{\compl}\cap\,\C{D}^{\compl})
\leq L_{k+1}^{-q+2}+\prob(\C{B}\cap\,\C{T}^{\compl}\cap\,\C{D}^{\compl}).
\]
Next, within the event
$\C{B}\cap\,\C{T}^{\compl}\cap \C{D}^{\compl}$, either
$\boxl_{L_{k+1}}^{(N)}(\B{u})$ or $\boxl_{L_{k+1}}^{(N)}(\B{v})$ is $E$-CNR.
It must be the FI-box $\boxl_{L_{k+1}}^{(N)}(\B{u})$. Indeed, by
Lemma \ref{lem:PITRoNS}, had box $\boxl_{L_{k+1}}^{(N)}(\B{v})$ been
both $E$-CNR and $(I,m)$-NPT, it would have been $(E,m)$-NS, which is not
allowed within the event $\C{B}$. Thus, the box $\boxl_{L_{k+1}}^{(N)}
(\B{u})$ must be $E$-CNR, but $(E,m)$-S. Hence,
\[
\C{B}\cap\,\C{T}^{\compl}\cap\,\C{D}^{\compl}\subset
\bigl\{\exists\;E\in I:\boxl_{L_{k+1}}^{(N)}
(\B{u})\text{ is }(E,m)\text{-S and }E\text{-CNR}\bigr\}.
\]
However, applying Lemma \ref{lem:Radial}, we see that
\begin{align*}
&\{\exists\, E\in I:\,\boxl_{L_{k+1}}^{(N)}(\B{u})
\text{ is }(E,m){\text{-S}}\text{ and }E{\text{-CNR}}\}\\
&\qquad\qquad\qquad\subset\{\exists\,
E\in I:\;\nu_{\rm S}(\boxl_{L_{k+1}}^{(N)}(\B{u});E)\geq 2\ell + \kappa(N)+1\}.
\end{align*}
Therefore,
\begin{align}  \label{eq:}        %\eqno(5.4)
\prob(\C{B}\cap\,\C{T}^{\compl}\cap\,\C{D}^{\compl})
&\leq\prob\accol{\exists\, E\in I:\;\nu_{\rm S}
(\boxl_{L_{k+1}}^{(N)}
(\B{u});E)\geq 2\ell + \kappa(N)+1 }\notag\\
&\leq 2L_{k+1}^{-1}\,L_{k+1}^{-2p}.
\end{align}
Finally, we get, with $q':=q/\alpha$,
\begin{align}  \label{eq:}        %\eqno(5.5)
\prob(\C{B})&\leq\prob(\C{B}\cap\,\C{T})
+\prob(\C{D})+\prob(\C{B}\cap\,\C{T}^{\compl}
\cap\,\C{D}^{\compl})\notag\\
&\leq\frac{1}{2}L_{k+1}^{-2p}+L_{k+1}^{-q'(N)+4}+2L_{k+1}^{-1}L_{k+1}^{-2p}
\leq L_{k+1}^{-2p},
\end{align}
for sufficiently large $L_0$, if we can guarantee, by taking $\eta>0$
small enough, that $q'(N)>2p+5$. This completes the proof of
Theorem \ref{thm:kplusoneforMI}.
\end{proof}
%-----------------------%

Therefore, Theorem~\ref{thm:MSAInd} is also proven.

%---------------------------------------------------------------------%
%:sec.8
%---------------------------------------------------------------------%
\section{Appendix. Proof of Lemmas \ref{lem:CondGeomSep}
and \ref{lem:PITRoNS}} \label{sec:appendix}

%-----------------------%
\begin{proof}[Proof of Lemma~\ref{lem:CondGeomSep}]
Given a positive integer $L$ a non-empty set $\C{J}\subset\{1,\ldots ,n\}$
and an $n$-particle vector $\B{y}=(y_1,\ldots ,y_n)
\in\D{R}^{nd}$ we say that the set of positions
$\accol{y_j}_{j\in\C{J}}$, forms an $(L+R)$-clump if the union
\begin{equation}  \label{eq:}        %\eqno(8.1)
\bigcup_{j\in\C{J}}\varLambda_R(x_j)\subset\D{R}^d
\end{equation}
yields a connected set.
Next, consider two
$n$-particle vectors $\B{x}$  and $\B{y}$ and proceed as follows.
\begin{enumerate}[1)]
\item
Decompose the vector $\B{y}$ into maximal $L$-clumps
$\varGamma_1$, $\ldots$, $\varGamma_M$ (of diameter
$\leq 2nL$ each), with the total number $M$ of clumps
being $\leq n$.
\item
To each position $y_i$ there corresponds precisely one clump,
$\varGamma_j$ where $j = j(i)\in\{1,\ldots ,M\}$.
\item
Suppose that there exists $j\in\{1,\dots,M\}$ such that
$\varGamma_j\cap\varPi\boxl^{(n)}_{L+R}(\B{x})=\varnothing$.
Then boxes $\boxl^{(n)}_L(\B{y})$ and $\boxl^{(n)}_L(\B{x})$
are separable.
\item
Suppose 3) is wrong; the aim is to deduce from the negation of 3)\
a necessary condition on possible locations
of vector $\B{y}$ and assess the number of
possible choices.
Indeed our hypothesis reads:
\begin{equation}  \label{eq:four.wrong}        %\eqno(8.3)
\varGamma_j\cap\varPi\boxl^{(N)}_{L+R}(\B{x})
\neq\varnothing\quad\text{for some }j=1,\dots,M.
\end{equation}
Therefore,
\[
\left.
\begin{array}{r}
\forall\;j=1,\dots,M,\;\exists\;i\text{ such that}\\
\abs{y_j-x_i}\leq
 \dist(y_j, \partial \varGamma_j) + \dist(\partial \varGamma_j, x_i)\\
\le [2n(L+R)-(L+R)]+L+R=2n(L+R)
\end{array}\!\!\right\}
\implies
\left\{\!\!
\begin{array}{l}
\forall\;j=1,\dots,M,\\
y_j\in\varPi\boxl^{(n)}_{A(L+R)}(\B{x})\\
\text{with }A\leq 2n.
\end{array}\right.
\]
We see that if a configuration $\B{y}$ is not separable from
$\B{x}$, then
every position $y_j$ must belong to one of the boxes
$\varPi_i\boxl^{(n)}_{AL}(\B{x})=\varLambda_{AL}(x_i)\subset\D{Z}^d$.
The total number of such boxes is $\leq n$. There are
at most $n^n$ choices of the boxes $\varLambda_{AL}(x_i)$ for
$n$ positions $y_1,\dots, y_n$; so we set $\kappa (n)=n^n$.
For any given choice among
$\le \kappa (n)$ possibilities, the point $\B{y} = (y_1,\dots, y_Nn)$
must belong to the Cartesian product of $n$ boxes of size $AL$,
i.e., to an $(nd)$-dimensional box of size $AL$.  The first assertion
of Lemma \ref{lem:CondGeomSep} now follows.

\item Next, consider a particular case where
$$
\boxl^{(n)}_{L+R}(\B{y}) \bigcap \boxl^{(n)}_{|\B{x}|+L+R}(\B{0})
= \varnothing.
$$
Then there exists at least one value of $i\in\{1, \ldots, N\}$ such that
\begin{equation}  \label{eq:independence}        %\eqno(8.4)
\varPi_i\boxl^{(n)}_{L+R}(\B{y}) \bigcap \varPi_i
\boxl^{(n)}_{|\B{x}|+L+R}(\B{0}) = \varnothing.
\end{equation}
However, by symmetry of the centered box $\boxl^{(n)}_{|\B{x}|+L+R}(\B{0})$
with respect to permutation of the coordinates, the projections
$\varPi_i \boxl^{(n)}_{|\B{x}|+L+R}(\B{0})$
are identical:
$$
\varPi \boxl^{(n)}_{|\B{x}|+L+R}(\B{0})
= \varPi_i \boxl^{(n)}_{|\B{x}|+L+R}(\B{0}),
\; i=1, \ldots, N.
$$
This implies separability of boxes
$\boxl^{(n)}_L(\B{y})$ and
$\boxl^{(n)}_{|\B{x}| + L}(\B{0})$.

This completes the proof of Lemma \ref{lem:CondGeomSep}.
\qedhere
\end{enumerate}
\end{proof}
%-----------------------%

We now pass to the proof of Lemma \ref{lem:PITRoNS}.
Recall that we consider an $n$-particle box of the form
$$
\boxl=\boxl_{L_{k+1}}^{(n)}(\B{u}),\quad\boxl'
=\boxl_{L_{k+1}}^{(n')}(\B{u}'),\quad\boxl''
=\boxl_{L_{k+1}}^{(n'')}(\B{u}''), \;
$$
with $\B{u} = (\B{u'}, \B{u''})\in\D{Z}^{nd}$, $\B{u'}\in\D{Z}^{n'd}$,
$\B{u''}\in\D{Z}^{n''d}$.
The corresponding Hamiltonian $\B{H}^{(n)}_{\boxl}$ has the following form:
$$
\B{H}^{(n)}_{\boxl} = \B{H}^{(n')}_{\boxl'} \otimes \one^{(n'')}_{\boxl''}
+ \one^{(n')}_{\boxl'} \otimes \B{H}^{(n'')}_{\boxl''}.
$$
Further, let $\{\BS{\varPsi}'_a, a\ge 1\}$ be normalized eigenfunctions
of $\B{H}^{(n')}_{\boxl'}$
and $\{E'_a, a\ge 1\}$ the corresponding eigenvalues. Correspondingly,
we denote by
$\{\BS{\varPsi}''_b, b\ge 1\}$ and $\{E''_b, b\ge 1\}$ (normalized)
eigenfunctions and eigenvalues of operator $\B{H}^{(n')}_{\boxl'}$. Then
the normalized eigenfunctions and respective eigenvalues of
$\B{H}^{(n)}_{\boxl}$ can be chosen in the form
$$
\B{\Psi}_{a,b} := \BS{\varPsi}'_a\otimes \BS{\varPsi}''_b,
\; E_{a,b} = E'_a + E''_b, \; a,b\ge 1.
$$
We assume that $E'_{a+1} \ge E'_a$, $E''_{b+1} \ge E''_b$, $a, b\ge 1$.

%Fix $j$ and consider the Green function
%$\green^{\boxl}(\B{v},\B{y};E_j)$, $\B{v},\B{y}\in\boxl$.

%-----------------------%
\begin{proof}[Proof of Lemma \ref{lem:PITRoNS}]
By hypothesis, $\boxl$ is $E$-CNR. Therefore,
$\forall\, a,b\ge 1$
$$
\begin{array}{l}
   e^{-L_{k+1}^\beta} < |E - E_{a,b}| = |E - (E'_a + E''_b)| \\
  \qquad\;\, = |(E - E'_a) - E''_b|   = |(E - E''_b) - E'_a |
\end{array}
$$
Therefore,
\begin{itemize}
\item
for all $E'_a$, the $n''$-particle box $\boxl''$ is $(E-E'_a)$-NR;
\item
for all $E''_b$, the $n'$-particle box $\boxl'$ is $(E-E''_b)$-NR.
\end{itemize}
%The same argument applies to smaller boxes
%$\boxl_{L_k}(\B{w'}) \times \boxl_{L_k}(\B{w''}) \subset \boxl$.
By the assumption of $(I,m)$-NPT, for all $E\in I$ the box $\boxl''$
should not contain two separable $(E-E'_a, m)$-S
sub-boxes of radius $L_k$. Therefore, the assumptions of
Lemma \ref{lem:RadialGF} hold true, and we deduce that the box
$\boxl''$ is $(E-E'_a)$-NS, yielding the required upper bound
for $\boxl''$.

The box $\boxl'$ is also $(I,m)$-NPT, by the hypothesis of the lemma,
so the same argument applies to $\boxl'$.

Let us now prove that box $\boxl$ is $(E,m)$-NS. If
$\B{v}=(\B{v}',\B{v}'')\in\partial\boxl_{L_{k+1}}^{(n)}(\B{u})$,
then either $\abs{\B{u}'-\B{v}'}=L_{k+1}$,
or $\abs{\B{u}''-\B{v}''}=L_{k+1}$.
First, consider the case where $\abs{\B{u}'-\B{v}'}=L_{k+1}$.
In this case  we can write the Green functions as
\begin{align}
\Green^{\boxl}(\B{u},\B{v};E)
&
=\sum_a\BS{\varPsi}'_a(\B{u}') \BS{\varPsi}'_a(\B{v}') \,
\sum_b\frac{\BS{\varPsi}''_b(\B{u}'') \BS{\varPsi}''_b(\B{v}'') }{
(E - E'_a) - E''_b }\notag\\
&= \sum_a \BS{\varPsi}'_a(\B{u}') \BS{\varPsi}'_a(\B{v}')
\green^{\boxl''}(\B{u}'', \B{v}''; E - E'_a).
\end{align}
For the resolvent operators we have the representation:
%\begin{align}
$$
\Green^{\boxl}(E)
= \sum_a \B{P}'_{\BS{\varPsi}'_a} \otimes \green^{\boxl''}(E - E'_a).
$$
%\end{align}
Here $\B{P}_{\BS{\varPsi}'_a}$ is the orthogonal
projection on the (normalized) eigenfunction $\BS{\varPsi}'_a$.
Naturally, $\| \B{P}_{\BS{\varPsi}'_a} \| = 1$.
 Recall that we aim to bound the norm
%$$
\begin{eqnarray}
   \|  \one_{\B{C}(\B{u} )} \Green^{\boxl}(E) \one_{\B{C}(\B{v} )}\|
   =\|  \one_{\B{C'}(\B{u'} )} \otimes \one_{\B{C''}(\B{u''} )}
   \; \Green^{\boxl}(E)  \one_{\B{C'}(\B{v'} )}
\one_{\B{C''}(\B{v''} )}\| \nonumber \\
   =\|  \one_{\B{C'}(\B{u'} )} \otimes \one_{\B{C''}(\B{u''} )}
   \; \Green^{\boxl}(E)  \one_{\B{C'}(\B{v'} )}
\one_{\B{C''}(\B{v''} )}\|  \nonumber \\
   \le \displaystyle
   \sum_a  \| \left( \one_{\B{C'}(\B{u'} )}
\B{P}_{\BS{\varPsi}'_a} \one_{\B{C'}(\B{v'} )} \right)
   \otimes
   \left( \one_{\B{C''}(\B{u''} )}
     \otimes \green^{\boxl''}(E - E'_a)
    \one_{\B{C''}(\B{v''} )}
    \right) \|    \nonumber \\
    \le \displaystyle
   \sum_a  \|  \one_{\B{C''}(\B{u''} )}
     \otimes \green^{\boxl''}(E - E'_a)
    \one_{\B{C''}(\B{v''} )} \|. \label{eq:tensor}
\end{eqnarray}
%$$

Since the interaction potential $U$ and the
external random potential $V(x;\omega )$ are non-negative,
the eigenvalues $E'_a$ satisfy $E'_a \ge E'_a(0)>0$ where
$E'_a(0)$ are the eigenvalues of the
operator $-\frac{1}{2}\B{\Delta}$ in the $n'$-particle
box $\boxl'$, by min-max principle.
Eigenvalues $E'_a(0)\nearrow \infty$ as $a\to\infty$, and
their growth rate is controlled by the Weyl formula.

This allows to perform an effective cut-off of the series in
the RHS of \eqref{eq:tensor}.
Namely, let $\delta>1$ be fixed, then the following quantity
is well-defined:
$$
A(\delta, \eta)
:= \max \{a\ge 1\,|: \; \eta - E'_a \ge -\delta  \}.
$$
Moreover,
$$
\begin{array}{c}
A(\delta, \eta )
\le C_{\rm{Weyl}}(\delta, n'd) |\B{B}^{(n')}_{L_{k+1}}(\B{u'})|  \\
E-E'_a \le E -E'_a(0).
\end{array}
$$
Here
$$
C_{Weyl}(n'd,\delta) =
\frac{\delta^{n'd/2}}{\Gamma\left(1+ \frac{n'd}{2}\right)
(4\pi)^{n'd/2}} < \delta^{n'd/2};
$$
so that we can use a more explicit upper bound
$A(\delta, \eta ) \le \delta^{n'd/2} |\B{B}^{(n')}_{L_{k+1}}(\B{u'})|$.

Further, for any $a \ge A(\delta,\eta)$ we have
$E-E'_a \le -\delta < 0$, so that the distance between
the point $E-E'_a$ and
the spectrum of operator $\B{H}_{\boxl''}$ is $>\delta$.
Then, by virtue of the Combes-Thomas estimate \footnote{For
small values
of the distance $\delta$, the decay exponent in the Green functions
is of order of $\sqrt{\delta}$, cf. \cite{BCH97}. Here the original
Combes--Thomas bound is stronger for large $|E-E'_a|$.}\cite{CT73},
$$
||\Green^{\boxl''}(E - E'_a)||<
e^{-c|E-E'_a|\,|u-v |}  < e^{-c\delta |u-v |}
$$
Now we chose $\delta$ large enough, thus making the exponent
$c\delta$ arbitrarily large. Taking into account the rate of
growth of $E'_a \ge E'_a(0)$, we can write
$$
\begin{array}{l}
\displaystyle \sum_{a > A_0(\delta,\eta)} \|  \one_{\B{C''}(\B{u''} )}
     \otimes \green^{\boxl''}(E - E'_a)
    \one_{\B{C''}(\B{v''} )} \| \\ \qquad\qquad
\le \displaystyle \sum_{a > A_0(\delta,\eta)}
e^{-c|E-E'_a| \|\B{u''} - \B{v''} \|}
\le  \displaystyle C_1 e^{-C_2 \delta L_{k+1}}
\end{array}
$$

Next, we have to estimate the norm of a finite sum
$$
\begin{array}{l}
\displaystyle \sum_{a=1}^{A_0(\delta,\eta)}
\|  \one_{\B{C''}(\B{u''} )} \otimes
\green^{\boxl''}(E - E'_a) \one_{\B{C''}(\B{v''} )} \| \\
\le \displaystyle A_0(\delta,\eta) \max_{1\le a \le A_0(\delta,\eta)}
\|  \one_{\B{C''}(\B{u''} )} \otimes
\green^{\boxl''}(E - E'_a) \one_{\B{C''}(\B{v''} )} \| \\
\le A_0(\delta,\eta) e^{-\gamma(m,L_{k+1},)}
\le \delta^{n'd/2} |\B{B}_{L_{k+1}}(\B{u'})| e^{-\gamma(m,L_{k+1})}
\end{array}
$$
where we used again that $\|\B{u''} - \B{v''} \| = L_{k+1}$.
Combining the two bounds, we obtain
$$
\begin{array}{l}
\displaystyle \sum_{a=1}^{\infty}
\|  \one_{\B{C''}(\B{u''} )} \otimes
\green^{\boxl''}(E - E'_a) \one_{\B{C''}(\B{v''} )} \| \\
\le \displaystyle
\delta^{n'd/2} |\B{B}_{L_{k+1}}(\B{u'})| e^{-\gamma(m,L_{k+1})}
+  C_1 e^{-C_2 \delta L_{k+1}} \\
\le 2\delta^{n'd/2} |\B{B}_{L_{k+1}}(\B{u'})| e^{-\gamma(m,L_{k+1})},
\end{array}
$$
for sufficiently large $L_0$ (hence, large $L_{k+1}$). Now
recall that the function $\gamma$ has the form
$$
\gamma(m,L,n) \;(\;= \gamma_N(m,L,n))\;\; = mL\left(1 + L^{-1/4} \right)^{N-n},
\; 1 \le n \le N,
$$
so that, for $n'\le n-1$, we have
$$
\gamma(m,L_{k+1},n') \ge \gamma(m,L_{k+1},N-1)
= mL_{k+1}\left(1 + L_{k+1}^{-1/4} \right)^{(N-n)+1}
$$
and
\begin{align*}
&-\ln \left( 2\delta^{n'd/2} |\B{B}_{L_{k+1}}(\B{u'})|
e^{-\gamma(m,L_{k+1},n-1)} \right) \\
&= mL_{k+1} \left(1 + L_{k+1}^{-1/4} \right)^{N-n+1} - C \ln L_{k+1} \\
&=  L_{k+1}\left(1 + L_{k+1}^{-1/4} \right)^{N-n+1}
\left(m -  CL_{k+1}^{-1} \ln L_{k+1} \right) \\
& \ge  L_{k+1}\left(1 + L_{k+1}^{-1/4} \right)^{N-n} m\,
\left(1 + L_{k+1}^{-1/4} \right)\left(1 -  L_{k+1}^{-1/2} \right)
\end{align*}
(provided that $L_{k+1}^{1/2} \ge Cm^{-1} \ln L_{k+1}$, which
is true for sufficiently large $L_0$)
\begin{align*}
&\ge  L_{k+1}\left(1 + L_{k+1}^{-1/4} \right)^{N-n} m,
\end{align*}
provided that $L_0>16$.
\par

Finally, note that in the case where $\abs{\B{u}''-\B{v}''}=L_{k+1}$,
we can use the representation
\begin{equation}
\green^{\boxl}(\B{u},\B{v};E)
=\sum_b\BS{\varPsi}''_b(\B{u}'') \BS{\varPsi}''_b(\B{v}'')\green^{\boxl'}(\B{u}',\B{v}';E-E''_b)
\end{equation}
and repeat the previous argument.
\qedhere
\end{proof}

%-----------------------%

%\newpage
%---------------------------------------------------------------------%
%:acknowledgements
%---------------------------------------------------------------------%
\section*{Acknowledgements}
VC thanks The Isaac Newton Institute (INI) and Department of Pure Mathematics and Mathematical Statistics,
University of Cambridge, for hospitality during visits in 2003, 2004,
2007 and 2008. YS thanks IHES, Bures-sur-Yvette, and
STP, Dublin Institute for Advanced Studies, for hospitality
during visits in 2003--2007. YS thanks the Departments of Mathematics
of Penn State University  and of UC Davis, for hospitality
during Visiting Professorships in the Spring of 2004, Fall of
2005 and Winter of 2008.  YS thanks the Department of Physics,
Princeton University and the Department of Mathematics of UC Irvine,
for hospitality during visits in the Spring of 2008. YS
acknowledges the support provided by the ESF Research Programme
RDSES towards research trips in 2003--2006.

%---------------------------------------------------------------------%
%:bib
%---------------------------------------------------------------------%

%---------------------------------------------------------------------%
%:end
%---------------------------------------------------------------------%
\end{document}